\documentclass[a4paper, 11pt]{amsart}
\usepackage[usenames]{xcolor}
\usepackage{xifthen}
\usepackage{algorithmic}
\usepackage{comment}
\usepackage{tcolorbox}
\usepackage{relsize}
\usepackage[font=footnotesize]{caption}
\usepackage[normalem]{ulem}
\usepackage{bm,bbm}
\usepackage{graphicx} 
\usepackage{float} 
\usepackage{subfigure} 

\def\fontsettingup{2} 

\usepackage{amsmath, amsthm, amssymb}
\usepackage[T1]{fontenc}
\ifthenelse{\fontsettingup = 1}{ \usepackage{eulerpx, eucal, tgpagella}   }{}
\ifthenelse{\fontsettingup = 2}{ \usepackage{mathpazo, tgpagella} }{}

\usepackage{hyperref}
\hypersetup{
  colorlinks=true,
  linkcolor=blue,
  filecolor=magenta,
  urlcolor=cyan,
  citecolor=blue
}

\usepackage[linesnumbered,boxed,ruled,vlined]{algorithm2e}
\usepackage{tikz}
\usepackage{cleveref}
\usepackage{makecell}

\usepackage[a4paper]{geometry}
\newgeometry{
textheight=9in,
textwidth=6.5in,
top=1in,
headheight=14pt,
headsep=25pt,
footskip=30pt
}

\newcommand{\jalaj}[1]{}
\newcommand{\liuexp}[1]{}

\newtheorem{theorem}{Theorem}

\newtheorem*{claim*}{Claim}

\newtheorem{fact}[theorem]{Fact}
\newtheorem{lemma}[theorem]{Lemma}

\newtheorem{corollary}[theorem]{Corollary}
\theoremstyle{definition}

\newtheorem{definition}[theorem]{Definition}
\newtheorem{remark}[theorem]{Remark}
\newtheorem*{remark*}{Remark}

\ifthenelse{\fontsettingup = 1}{
  \def\*#1{\mathbf{#1}} 
  \def\+#1{\mathcal{#1}} 
  \def\-#1{\mathrm{#1}} 
  \def\^#1{\mathbb{#1}} 
  \def\!#1{\mathfrak{#1}} 
}{}

\ifthenelse{\fontsettingup = 2}{
  \def\*#1{\boldsymbol{#1}} 
  \def\+#1{\mathcal{#1}} 
  \def\-#1{\mathrm{#1}} 
  \def\^#1{\mathbb{#1}} 
  \def\!#1{\mathfrak{#1}} 
}{}

\DeclareMathOperator*{\argmin}{arg\,min}

\def\oE{\mathbb{E}}
\newcommand{\E}[2][]{ \ifthenelse{\isempty{#1}}
  {\oE\left[#2\right]}
  {\oE_{#1}\left[#2\right]} }

\DeclareMathOperator*{\oVar}{\mathbf{Var}}
\newcommand{\Var}[2][]{ \ifthenelse{\isempty{#1}}
  {\oVar\left[#2\right]}
  {\oVar_{#1}\left[#2\right]} }

\def\oEnt{\mathbf{Ent}}
\newcommand{\Ent}[2][]{ \ifthenelse{\isempty{#1}}
  {\oEnt\left[#2\right]}
  {\oEnt_{#1}\left[#2\right]} }

\renewcommand{\epsilon}{\varepsilon}

\newcommand{\lap}{\texttt{Lap}}

\newcommand{\todo}[1]{}

\newcommand{\ju}[1]{}
\newcommand{\george}[1]{}
\newcommand{\zongrui}[1]{}

\let\epsilon=\varepsilon
\newcommand{\citet}{\cite}
\newcommand{\citep}{\cite}

\usepackage[foot]{amsaddr}

\makeatletter
\renewcommand{\email}[2][]{%
  \ifx\emails\@empty\relax\else{\g@addto@macro\emails{,\space}}\fi%
  \@ifnotempty{#1}{\g@addto@macro\emails{\textrm{(#1)}\space}}%
  \g@addto@macro\emails{#2}%
}
\makeatother

\title{Almost linear time differentially private release of synthetic graphs}
\author{Jingcheng Liu}
\address{State Key Laboratory for Novel Software Technology and New Cornerstone Science Laboratory, Nanjing University}
\email{liu@nju.edu.cn}
\author{Jalaj Upadhyay}
\address{Rutgers University}
\email{jalaj.upadhyay@rutgers.edu}
\author{Zongrui Zou}
\address{State Key Laboratory for Novel Software Technology and New Cornerstone Science Laboratory, Nanjing University}
\email{zou.zongrui@smail.nju.edu.cn}
\date{}

\begin{document}
\pagenumbering{arabic}

\begin{abstract}
  In this paper, we give an almost linear time and space algorithms to sample from an exponential mechanism with an $\ell_1$-score function defined over an exponentially large non-convex set. As a direct result, on input an $n$ vertex $m$ edges graph $G$, we present the \textit{first} $\widetilde{O}(m)$ time and $O(m)$ space algorithms for differentially privately outputting an $n$ vertex $O(m)$ edges synthetic graph that approximates all the cuts and the spectrum of $G$. These are the \emph{first} private algorithms for releasing synthetic graphs that nearly match this task's time and space complexity in the non-private setting while achieving the same (or better) utility as the previous works in the more practical sparse regime. Additionally, our algorithms can be extended to private graph analysis under continual observation. 
\end{abstract}
\maketitle

\section{Introduction}
Consider a (hypothetical) for-profit organization that has the following two use cases: (i) The users can purchase (or download) smartphone applications in their online store. 
(ii) The users can download music albums or subscribe to podcasts in their music store. 
Both these relationships can be modeled by a sparse graph, where the unweighted graph encodes whether or not a user downloaded an app or music. Such graphs  (in fact, for graphs encoding most practical use cases, also see Appendix~\ref{app:usecases}) are usually sparse, and unweighted versions of these graphs are not useful for downstream tasks, and might already be known to the organization without a privacy umbrella.  
Unweighted sparse graphs are also not that interesting in practice for reasons we discuss in more detail in \Cref{sec:introcutapprox}.

The more interesting graphs in practice are weighted graphs. For example, it stores information on the frequency of app usage, which app is used more often at a different time of the day, resources an app uses, and in conjunction with which app. Similarly, the more interesting data for a music app is whether or not a user listens to a song (and how frequently) in an album, what part of the podcast is played more often to understand engagement, etc.  Analysis of such graphs is useful for content creators and app developers to develop a better application or create content that provides more engagement. These analyses are naturally cast as cut functions on these graphs. 

However, providing this information without a robust privacy guarantee can lead to serious privacy concerns and even  legal repercussions. To assuage these (future) leakages of information, a natural candidate is to answer these cut functions using {\em differential privacy}~\citep{dwork2006calibrating}. Parameterized by privacy parameters $(\epsilon,\delta)$, differential privacy guarantees that the output distribution of an algorithm is not sensitive to small changes in the input dataset, and has inspired a lot of further research and practical applications in industries and government agencies (see Desfontaines' blog \cite{desfontainesblog20211001} for an up-to-date list). 

Differentially-private graph approximation is one of the most widely studied problems in differential privacy. Even though there are many differentially private algorithms for releasing synthetic graphs preserving combinatorial and algebraic properties of the input graph, these algorithms have no real-world implementations. To understand this discrepancy between the theory and practice, consider the current state-of-the-art differentially private algorithm for answering cut functions for unweighted~\citep{eliavs2020differentially} and weighted graphs~\citep{liu2023optimal}\footnote{Similar or closely related issues have been pointed out to us for other known algorithms by practitioners.}.  Both algorithms are based on a private mirror descent and run in time $\widetilde O(n^7)$, use $O(n^2)$ space, and ensure an expected error of $O(\sqrt{mn}\log^2(n/\delta))$ for constant $\varepsilon$\footnote{Spectral approximation algorithms also suffer from similar limitations of large space and running time.}.  Here $n$ is the number of nodes and $m$ is the number of edges. Even though it matches the lower bound for a constant probability of error, this algorithm suffers from the following main limitations: 
\begin{enumerate}
    \item Real-world graphs are usually sparse, high edge weighted with $n$ in the orders of $10^9$ (see \citet{goswami2021sparsity, bellingeri2023considering} and  \Cref{app:usecases}). Therefore, generating a synthetic graph using the private mirror descent is practically infeasible due to the $\widetilde{O}(n^7)$ running time.
    
    \item The error bounds given by private mirror descent are in expectation \citep{eliavs2020differentially, liu2023optimal}. In contrast, developers and creators prefer high probability bounds (or confidence intervals) as fluctuation in output is hard to interpret. One can get a high probability bound by Markov's inequality for both these algorithms but at the cost of significant degradation in accuracy as verified by the authors of Eli{\'a}{\v{s}} et al. \cite{eliavs2020differentially}. 
    
    \item The synthetic graphs given by the known state-of-the-art algorithms are dense and do not preserve any combinatorial structures like {\em sparsity}, making subsequent analysis on them significantly more time-consuming.  
\end{enumerate}  
Due to these limitations, profit-driven organizations prefer non-private algorithms for processing graph data, which damages user privacy. 
This motivates  the central thesis of this paper: 
\begin{quote}
    Design differentially private algorithms to release a synthetic graph that approximates the spectrum and cut functions of the graph. Further, (i) the computational time required by the curator and the analyst should be as close to the non-private setting, and (ii) the output graph should preserving the sparsity.
\end{quote}

Since real-world graphs are sparse and answering cut queries on unweighted sparse graphs is often not that interesting in practice (also see \Cref{sec:introcutapprox}), we focus on sparse-weighted graphs. 
Our key technical contribution that underpins fast private approximation on graphs is a sampling algorithm shown in Appendix \ref{sec:proof_main}:

\begin{lemma}[Informal]
\label{lem:main}
     Fix any $m\leq N$. Given a non-negative $N$-dimensional real vector $x \in \mathbb R^N_{\geq 0}$ with sparsity $\|x\|_0 = |\{i: x_i\neq 0\}| \leq m$, let $\pi$ be the distribution such that $ \pi [S] \propto e^{-\epsilon\|x-x|S\|_1}$ with support $\{S\in \{0,1\}^N : \|S\|_0 = m \}$.
    Then, there is an algorithm that approximately samples from $\pi$ in  $\widetilde O(m)$ time.
\end{lemma}
Here, $x|S$ is the restriction of $x$ on $S$. That is, for any $ i \in [N]$, $(x|S)_i = x_i$ if $S_i = 1$ and $(x|S)_i = 0$ otherwise. In particular, if $S^*$ is the support of $x$, then $x = x|S^*$, and thus $S^*$ maximizes the probability of being chosen by $\pi$. Our approach is the first sampling perspective to a problem where optimization perspective has been predominately used~\cite{eliavs2020differentially, gupta2012iterative,liu2023optimal}. Such a study has been very fruitful in the theory of optimization~\cite{chewi2023optimization,chewi2024log, ganesh2022langevin,gopi2022private}. 
We explore some applications of Lemma~\ref{lem:main} below:

\begin{enumerate}
    \item \textbf{Spectral approximation in linear time.} One important but difficult task in the context of differential privacy is to approximately preserve the spectrum of the graph Laplacian. This is studied in both  non-private~\citep{ allen2015spectral,batson2012twice,lee2015sparsified, spielman2011graph} and private setting~\citep{arora2019differentially, blocki2012johnson, dwork2014analyze, upadhyay2021differentially}. Using Lemma \ref{lem:main}, we show a linear time algorithm that takes an $n$ vertices $m$ edges graph of maximum {\em unweighted} degree $d_{\mathsf{max}}$, and outputs the spectral approximation with a purely additive error of ${O}(d_{\mathsf{max}}\log(n/\delta))$. 
    This improves the error rate in Liu et al.~\cite{liu2023optimal} by $\log(n)$ factor, the running time from $O(n^2m)$ to $O(m)$. It also improves Dwork et al. \cite{dwork2014analyze} in the more practical setting, $d_{\mathsf{max}} =\widetilde{o}(\sqrt{n})$.

    \item \textbf{Cut approximation in linear time.} Given a graph with $n$ vertices and $m$ weighted edges, we give {\em the first $\widetilde{O}(m)$ time algorithms that output a synthetic graph with at most $m$ edges.} This implies that an analyst can compute the answer to any cut query in $O(m)$ time, resolving issues 1 and 3. Our algorithms also achieve an almost linear additive error of $3m\log(n/\delta)$ with probability $1-o(\delta)$, which means better accuracy in the sparse regime (see Table \ref{t.1} and Section \ref{s.results}), resolving issue 2. We give a more detail comparison in \Cref{sec:introcutapprox}.  
    
    \item \textbf{Continual release algorithm.} Our linear time algorithms can be used to construct an algorithm for efficiently output a stream of synthetic graphs under continual observation~\citep{chan2011private, dwork2010differentially} while preserving {\em event-level} privacy. In comparison, combining the previous state-of-the-art algorithms \citep{eliavs2020differentially, liu2023optimal} and the standard transformation \citep{chan2011private, dwork2010differentially}, one could get continual release algorithms with $O(n^2)$ update time (and $\widetilde O(n^{3/2})$ additive error) or $O(n^7\log n)$ update time (and $O(\sqrt{nm})$ additive error). In section \ref{sec:contious_observation}, we show how to reduce the update time  at each round to only $O(\log n)$, if the number of updates, $T$, is polynomial in $n$.
\end{enumerate}

\subsection{Problem definition.}
 We consider the class of undirected positively weighted graph defined over $n$ vertices. We set $N = {n\choose 2}$ throughout this paper. We represent a graph $G =(V,E)$ by a $N$-dimensional vector encoding the edge weights, i.e., $G \in \mathbb R_{\geq 0}^N$, where the $e$-coordinate, $G[e] = w_e\in \mathbb{R}_{\geq 0}$, is the weight of the edge, $e \in E$, for a canonical ordering on the edges. Throughout the paper, we let $E = \{1\leq e\leq N |w_e>0\}$ denote the edge set of $G$. We let $[n]$ denote the set of integers $\{1,2,\cdots,n\}$. We use differential privacy as the measure of privacy loss, and the standard notion of neighboring for graphs known as \textit{edge level differential privacy}~\citep{arora2019differentially, blocki2012johnson, dwork2014algorithmic,  eliavs2020differentially, hardt2012beating, gupta2012iterative, upadhyay2021differentially,upadhyay2013random}: two graphs are \emph{neighboring} if their weights differ in one pair of vertices by at most 1. 
\begin{definition}[Differential privacy  \citep{dwork2006calibrating}]
\label{d.dp}
   Fix any  $\epsilon>0$ and $\delta\in [0,1)$. Let $\mathcal{A}:\mathbb{R}_{\geq 0}^{{n\choose 2}}\rightarrow \mathcal{R}$ be an algorithm mapping graphs to any output domain $\mathcal{R}$. Then $\mathcal{A}$ is {\em $(\epsilon,\delta)$-differentially private} (DP) if for any pair of neighboring graphs $G, G'$ and any output event $S\subseteq \mathcal{R}$, it holds that
  $Pr[\mathcal{A}(G)\in S]\leq e^\epsilon \cdot Pr[\mathcal{A}(G)\in S] + \delta.$   
  In particular, if $\delta=0$, we simply say $\mathcal{A}$ is $\epsilon$-DP.
\end{definition}

We note that in Definition \ref{d.dp}, the edge set (i.e., the topology) of two neighboring graphs can be different. For example, consider neighboring graphs $G$ and $G'$ where $G'$ has an extra edge $e$ with weight $G'[e]= 1$, while $G[e] = 0$. 
We do not assume that the number of edges, $|E|$, is public information. 
However, all our algorithms can be easily adapted if the number of edges is publicly known as in previous works~\citep{arora2019differentially, blocki2012johnson, eliavs2020differentially, gupta2012iterative, upadhyay2021differentially}. 

\paragraph{Utility metric on spectral approximation} Given a graph $G$, its Laplacian is defined as $L_G:= D_G - A_G$, where $A_G$ is the weighted adjacency matrix and $D_G$ is the weighted degree matrix. Then the goal is to output a graph, $\widehat{G} = \argmin_{\widehat G}\left\|L_G - L_{\widehat{G}}\right\|_2$. Here  $\|\cdot\|_2$ denotes the spectral norm.

\textbf{Linear queries on a graph.} One relaxation of spectral approximation is {\em linear queries} on a graph. 
In this problem, we view a graph $G$ as a vector in $\mathbb{R}_{\geq 0}^{N}$, and the aim to find a $\widehat{G} \in \mathbb{R}^{N}$ that minimizes
$
Eval(G,\widehat{G}) = \max_{q\in [0,1]^N} |q^\top G - q^\top \widehat{G}|.
$

\paragraph{Utility metric on cut approximation} 

A special case of answering linear queries on graphs is \textit{cut approximation}, where we aim to find a $\widehat{G}$ privately such that for all disjoint $S,T\subseteq V$, it holds that 
\begin{align}
 \label{eq4}
\begin{split}
        \left|\Phi_G(S,T) -  \Phi_{\widehat{G}}(S,T)\right| \leq  \xi,  
    \text{ where } \Phi_G(S,T) = \sum_{e=(u,v) \in S \times T} G[e].
\end{split}    
\end{align}
Here, the goal is to minimize $\xi$~\citep{dwork2006calibrating, dwork2014algorithmic,   eliavs2020differentially,  gupta2012iterative, liu2023optimal,upadhyay2021differentially}. 
We refer to an error of this form \textit{purely additive error}. Some previous works also require a \emph{multiplicative error} in addition to the additive error\footnote{The bound using exponential mechanism \citet{mcsherry2007mechanism} is shown in Section 3.2.4 in Blocki et al. \citet{blocki2012johnson}.}~\citep{arora2019differentially, blocki2012johnson, mcsherry2007mechanism,upadhyay2013random}.
All our algorithms achieve purely additive error.  For brevity, we write  $(S,V\backslash{S})$-cuts as  $\Phi_G(S)=\Phi_G(S,V\backslash S)$.

\subsection{Overview of our results}\label{s.results} \label{sec:introcutapprox}

Based on the efficient sampler for exponential mechanisms with $\ell_1$ norm scoring function under the sparsity constraint (Lemma \ref{lem:main}), we give a brief overview of our results on efficient private graph approximation in various settings.

\textbf{Private spectral approximation in linear time and space.}\label{s.intro_spectral}
One of our linear time and linear space algorithms (Algorithm \ref{alg2}) also preserves the spectrum of $L_G$ for graphs of maximum unweighted degree $d_{\mathsf{max}} \leq n-1$. For a weighted graph, the {\em unweighted} degree of a vertex $u\in [n]$ is the number of (weighted) edges incident to $u$, instead of the total sum of their weights. For spectral approximation, we improve Liu et al.~\cite{liu2023optimal} by a $\log(n)$ factor and also improve on their run time by at least an $\Omega(n^2)$ factor: 
\begin{table*}[t]
    \centering
    \caption{The comparison of existing $(\epsilon,\delta)$-differentially private algorithm on spectral approximation.}\label{table.spectral}
    \begin{tabular}{|c|c|c|c|c|}
        
        \hline
        \textbf{Method} &  \makecell[c]{\textbf{Additive error on} \\\textbf{spectral approximation}} &  \makecell[c]{\textbf{Preserve} \\\textbf{sparsity?}} & \makecell[c]{\textbf{Purely }\\\textbf{Additive?}} & Run-time \\
        \hline

        JL mechanism \cite{blocki2012johnson} & ${{O}}\left( \frac{\sqrt{n}\log(n/\delta)}{\epsilon}\right)$ & No  &  No & $O(n^3)$ \\
        \hline
        
        Analyze Gauss \cite{dwork2014algorithmic} & ${{O}}\left( \frac{\sqrt{n}\log(n/\delta)}{\epsilon}\right)$ & No  & Yes & $O(n^2)$ \\
        \hline 
        
        Liu et al. \cite{liu2023optimal} & ${{O}}\left( \frac{d_{\mathsf{max}}\log^2(n)}{\epsilon}\right)$ & Yes  & Yes & $O(n^2|E|d_{\mathsf{max}})$ \\
        \hline 
 
        {\color{red}This paper} & {\color{red}${O}\left(\frac{d_{\mathsf{max}}\log(n/\delta)}{\epsilon}\right)$} &  \textbf{Yes} & \textbf{Yes} & $\widetilde{O}(|E|)$\\
        \hline 
  
    \end{tabular}
  \end{table*}

\begin{theorem}
    [Informal version of Theorem \ref{t.spectral}]
    Given privacy parameters $\epsilon>0$,  $0<\delta<1$ and a graph $G$, there is a linear time and space  $(\epsilon,\delta)$-differentially private algorithm that outputs a graph  $\widehat{G}$ such that, with high probability, 
    $ \left\|L_G - L_{\widehat{G}}\right\|_2 = O\left(\frac{d_{\mathsf{max}}\log(n/\delta)}{\epsilon}\right).$
\end{theorem}

The main idea behind the proof is to show that for maximum unweighted degree $d_{\mathsf{max}}$, our algorithm actually preserves \text{\em every} normalized $(S,T)$-cuts (not only the maximum one) with
a small purely additive error. We use this fact to bound the spectral radius of $L_G - L_{\widehat{G}}$.  We compare this new upper bound with previous results in Table \ref{table.spectral}. As evident from the table, we only incur an additive error and have the best utility when $d_{\mathsf{max}} = o(\sqrt{n})$ and $\delta = n^{-O(1)}$.

\textbf{Private cut approximation in linear time and space with linear error.}
We also give two linear time algorithms that output a synthetic graph that incurs linear error for sparse graphs.  That is, we show the following theorem on differentially privately cut approximation:
\begin{theorem}[Informal version of \Cref{t.ana_on_alg2} and \Cref{t.small_cut_alg2}]\label{t.intro_graph_informal}
    Given $\epsilon>0, \delta \in (0,1)$ and non-negative weighted $n$ vertices $m$ edges graph $G \in \mathbb R^{n \choose 2}_{\geq 0}$ with edge set $E$ and maximum unweighted degree, i.e., the maximum number of edges incident on any vertex being $d_{\mathsf{max}}$. There is an $\widetilde{O}(|E|)$ time, ${O}(|E|)$ space, $(\epsilon,\delta)$-differentially private algorithm that outputs a synthetic graph $\widehat{G}$ with $O(|E|)$ edges, such that, with high probability, $\forall S,T\subseteq V \text{ such that } S\cap T = \varnothing ,$
    $$|\Phi_{G}(S,T) - \Phi_{\widehat{G}}(S,T)| \leq \min\Big\{3|E| , 4d_{\mathsf{max}}|S|, 4d_{\mathsf{max}}|T|\Big\}\frac{\log (2n/\delta)}{\epsilon}.$$
\end{theorem}

The known lower bound \cite{eliavs2020differentially} implies that, if $m=\widetilde O(n)$, any $(O(1),\delta)$-differentially private algorithm for answering all cut queries incurs an error $\widetilde \Omega(n)$. In contrast, just outputting a trivial graph incurs an $\widetilde O(n)$ error if the average weight is $\widetilde O(1)$. Therefore, from a theoretical perspective, the interesting case is when the average weight is $\widetilde \omega(1)$.

In addition to using minimal resources and being the first almost linear time algorithms, our algorithms also improve on the state-of-the-art algorithm~\citep{eliavs2020differentially,liu2023optimal} in terms of error bound under various settings:
\begin{enumerate}
    \item Both Eli{\'a}{\v{s}} et al. \cite{eliavs2020differentially} and Liu et al. \cite{liu2023optimal} give the bound in expectation. When considering the high probability bound (even for a mild success probability of at least $1-\sqrt{1/n}$), we improve on Eli{\'a}{\v{s}} et al. \cite{eliavs2020differentially} and Liu et al. \cite{liu2023optimal} irrespective of whether the graph is sparse or dense and whether it is weighted or unweighted (also see \Cref{rem:introhighprob})\footnote{The authors of Eli{\'a}{\v{s}} et al. \cite{eliavs2020differentially} verified this. Since Liu et al. \cite{liu2023optimal} uses Eli{\'a}{\v{s}} et al. \cite{eliavs2020differentially} as a subroutine, it also applies to Liu et al. \cite{liu2023optimal}.}. The probability requirement here is very mild; in practice, we often require a probability of success to be at least $1-1/n$. 

    \item For $\epsilon=\Theta(1)$, any weighted graph with average weight $W = \widetilde{\omega}(1)$, the expected error bound in previous works~\cite{eliavs2020differentially,liu2023optimal} is $O(\sqrt{Wmn}\log^{2}(n/\delta))$ and $O(\sqrt{mn}\log^{3}(n/\delta))$, respectively. That is, when the error bound is guaranteed with a constant probability of error, we improve on \cite{eliavs2020differentially} whenever the average degree is $o(W\log^2(n/\delta))$ and \cite{liu2023optimal} whenever the average degree is $o(\log^4(n/\delta))$. We note that naturally occurring graphs have average degree $o(\log^2(n/\delta))$ (see Appendix~\ref{app:usecases} for a more detailed description).

\end{enumerate}

\begin{remark}
    [Discussion regarding high probability bound] 
\label{rem:introhighprob}
The bound in Eli{\'a}{\v{s}} et al. \cite{eliavs2020differentially} is in expectation and they refer to \cite{nemirovski2009robust} to get a high probability bound. \cite{nemirovski2009robust} presents two general ways to get a high probability bound from the expectation. The stronger condition (stated as eq. (2.50) in Nemirovski et al. \cite{nemirovski2009robust} and stated below) allows a high probability bound (i.e., with probability $1-\beta$) while incurring an extra $\log(1/\beta)$ factor if $\mathbb E\left[\exp\left({\|g\|_\infty^2 \over M^2 }\right)\right] = \int n^2 \cdot exp(-t) \cdot exp(t^2/M^2)$ 
is bounded. Here $\|g\|_\infty$ is the $\ell_\infty$ norm of gradients in the mirror descent. However, the algorithm in Eli{\'a}{\v{s}} et al. \citet{eliavs2020differentially} does not satisfy this stronger condition; the expectation is unbounded if $M$ is poly-logarithmic in $n$. Eli{\'a}{\v{s}} et al. \cite{eliavs2020differentially} satisfies the weaker condition (stated as eq. (2.40) in Nemirovski et al. \cite{nemirovski2009robust}); Lemma 4.5 in Eli{\'a}{\v{s}} et al. \cite{eliavs2020differentially} implies that $\mathbb E[\|g\|_\infty^2] = O(\log^2 n).$ This condition though translates an expectation bound to $1-\beta$ probability bound with an extra ${1 \over \beta^2}$ factor in the error. In particular, if we want $\beta = o(1/\sqrt{n})$, then Eli{\'a}{\v{s}} et al. \cite{eliavs2020differentially} results in an error $O(\sqrt{mn^3/\epsilon} \log^2(n/\delta))$. This is worse than our bound, irrespective of whether the graph is weighted, unweighted, or sparse or dense. 
\end{remark}

\begin{table*}[t]
    \centering
    \caption{The comparison of existing results on cut approximation (for sparse case, $|E| =\widetilde O(n)$).}\label{t.1}
    \begin{tabular}{|c|c|c|c|c|}
        
        \hline
        \textbf{Method} &  \makecell[c]{\textbf{Additive error for} \\$(S,V\backslash S)$ \textbf{cuts}} &  \makecell[c]{\textbf{Output a }\\\textbf{sparse graph?}} & \makecell[c]{\textbf{Purely }\\\textbf{additive?}} & Run-time \\
        \hline 
        Exponential mechanism \cite{mcsherry2007mechanism} &${{O}}\left(\frac{n\cdot\log n}{\epsilon} \right)$ & Yes  & No & Intractable \\
        \hline 
        
        JL mechanism \cite{blocki2012johnson} & ${{O}}\left(\frac{n^{1.5}\cdot\texttt{poly}(\log n)}{\epsilon}  \right)$ & No  &  No & $O(n^3)$ \\
        \hline
        
        Analyze Gauss \cite{dwork2014analyze} & ${{O}}\left(\frac{n^{1.5}\cdot\texttt{poly}(\log n)}{\epsilon}   \right)$ & No  & Yes & $O(n^2)$ \\
        \hline 
        
        Mirror descent \cite{eliavs2020differentially,gupta2012iterative} & ${{O}}\left(\frac{n\sqrt{W} \texttt{poly}(\log n)}{\epsilon^{1/2}}\right)$ &No &  Yes & $\widetilde O(n^7)$ \\
        \hline 
        
        Liu et al. \cite{liu2023optimal} & ${{O}}\left(\frac{n\cdot  \texttt{poly}(\log n)}{\epsilon}\right)$ &No &  Yes & $\widetilde O(n^7)$ \\
        \hline         
        {\color{red}This paper} & {\color{red}\textbf{${{O}}\left(\frac{n\cdot\texttt{poly}(\log n)}{\epsilon} \right)$}} &  \textbf{Yes} & \textbf{Yes} & $\widetilde{O}(n)$\\
        \hline 
  
    \end{tabular}
  \end{table*}

We give a detailed comparison with other works in Table~\ref{t.1}. We only compare the results w.r.t. the error on $(S,V\backslash{S})$-cuts instead of $(S,T)$-cuts, 
 since two major results~\citep{blocki2012johnson, mcsherry2007mechanism} do not provide guarantees on $(S,T)$-cuts. 
 If an algorithm preserves all $(S,V\backslash{S})$-cuts with a purely additive error, then it also preserves all $(S,T)$-cuts with the same error because,  for any disjoint $S, T \subseteq V$, $\Phi_G(S,T)$ can be computed using $\Phi_G(S,V \setminus S) + \Phi_G(T,V \setminus T) - 2\Phi_G(S\cup T,V \setminus (S\cup T)).$ We also comprehensively summarize by comparing previous techniques in Appendix \ref{s.compare}.

\begin{remark} [On the scale invariance] \label{rem:scale}
     Note that the work of Eli{\'a}{\v{s}} et al. \cite{eliavs2020differentially} showed a $\sqrt{\epsilon^{-1}}$ dependency on $\epsilon$ for unweighted graph, while the dependency for weighted graph is $\epsilon^{-1}$ was recently shown in Liu et al. \cite{liu2023optimal}. Indeed, if we have a dependency on $W$ better than $\sqrt{W}$, we cannot expect to have a better dependency on $\epsilon$ than $\sqrt{\epsilon^{-1}}$, otherwise, we can scale up the weight and scale down after obtaining the answer, then we will get a more accurate approximation, which violates the scale invariance.
\end{remark}

\textbf{Optimally answering linear queries.}
One natural relaxation of spectral approximation is answering {\em linear queries} on a graph, where we view an $m$ edges graph as an ${n \choose 2}$ dimensional vector formed by the edge-weights, $G \in \mathbb R_{\geq 0}^{n \choose 2}$. We show that our algorithm can actually outputs a synthetic graph privately that can be used to answer any linear query formed by a vector $q \in [0,1]^{n \choose 2}$ with error $ \widetilde{O}\left({m }\right)$ in linear time. {Note that $(S,T)$-cut queries are special cases of linear queries $q\in [0,1]^N$ as it can be identified as a linear query $q_{(S,T)} \in \{0,1\}^{{n \choose 2}}$ whose $e$-th entry, 
$q_{(S,T)}[e]$ is $1$ only if the edge $e=(u,v) \in S \times T$.} By a packing argument, we also show that 
$\widetilde{O}(|E|)$ error is indeed optimal for answering all possible linear queries:

\begin{theorem}[Informal version of  \Cref{t.lb_on_sparse}]\label{t.intro_lb_on_sparse}
Fix $\epsilon>0$ and $0<\delta<1$. 
Let $\mathcal{M}:\mathbb{R}_{\geq 0}^N\rightarrow \mathbb{R}^N$ be an $(\epsilon, \delta)$ differentially private algorithm. 
Then, there exists a graph $G$, viewed as a vector in $\mathbb{R}^N$, such that
$\mathbb{E}_{\mathcal M}[Eval(G,\mathcal{M}(G))] = \Omega\left(\frac{(1-\delta)|E|}{e^{\epsilon}+1}\right).$ 
\end{theorem}

\subsubsection{Extension to continual observation.}
In Section \ref{sec:contious_observation}, we develop an efficient framework for private cut approximation under continual observation using the {\em binary mechanism}~\citep{chan2011private} and recent improvements~\citep{chen2017pegasus, cao2018quantifying, henzinger2022constant,henzinger2023almost, henzinger2023unifying,jain2021price}. The main barrier in practice that prohibits transforming previous static graph algorithms into the setting under continual observation is the update-time overhead. For example, in the binary mechanism, we have to resample a synthetic graph from the partial sum of some updates at every round in $[T]$. This means that the running time per iteration will be at least $O(n^2)$ for previous static algorithms on cut approximation \citep{dwork2014analyze}. 
We significantly reduce the update-time overhead with our algorithms, thanks to their sparsity-preserving property and linear running time. In the binary mechanism, given $T\in \mathbb{N}_+$ (the number of updates), there are at most $\lceil \log_2 T \rceil$ layers, and the $i$-th layer (from bottom to top) contains $T/2^i$ graphs where each graph has at most $2^i$ edges (each update is a leaf in the computation tree). By the time guarantee of our linear time algorithms (see Algorithm \ref{alg2}, for example), outputting a graph with $m$ edges needs time $O(m)$. Thus, to compute all synthetic graphs corresponding to the nodes in all layers, the running time will be
$\sum_{i=0}^{\lceil \log_2 T\rceil} O(2^i) \cdot {T \over 2^i} = O(T\log_2 T).$ 
Since there are $T$ rounds in total, then the amortized run-time for each round is $O(\log T)$. 
In particular, we show that if $T = \text{poly}(n)$, then there is a $(\epsilon,\delta)$-differentially private algorithm under continual observation such that with only $O(\log n)$ amortized run-time each round, it approximates the size of all $(S,T)$-cuts at each round with error at most $\widetilde{O}(m/\epsilon\cdot (\log T\log (n/\delta))^{1.5})$, where $m$ is the number of edges in the final state.

\section{Technical Overview}\label{app:tech_overview}
Here, we discuss the technical ingredients behind our linear time algorithms in detail. Our two algorithms that achieve additive error in linear run-time come from simple intuitions described next:  (i) efficient sampling from the exponential mechanism and (ii) high pass filter.

\paragraph{(1) Efficient sampling from exponential mechanism} Our first sampling-based algorithm is based on simulating the exponential mechanism with a proper score function and is defined over a non-convex set. To the best of our knowledge, {\em this is the first algorithm that simulates the exponential mechanism over a non-convex set in almost linear time and linear space} and might be of independent interest. 

For the exposition, first, we consider an input graph $G \in \mathbb{N}^{n\choose 2}$ with $m$ edges, each with a non-negative integer weight bounded by $W\in \mathbb{N}$. If $m$ and $W$ are publicly known, then there are at most $Q = {N \choose m}\cdot W^{m}$ (where $N = {n\choose 2}$) ways to re-arrange edges and their weights, each of them is a possible output. We consider each of them as a candidate, and sample one of the candidates according to the exponential mechanism, with the scoring function being the maximum error in  $(S,T)$-cuts. This scoring function is $1$-Lipchitz, so privacy follows from that of the exponential mechanism. The utility follows from the standard analysis and the error is $O(\log (Q)) = O(m\cdot \log(Wn^2))$~\citep{dwork2014algorithmic}. For graphs in $\mathbb{R}_{\geq 0}^N$ with real-valued weights, a continuous version of the exponential mechanism gives a similar utility bound.

The above algorithm has two issues: it is computationally infeasible and the error has a dependency on $W =\|G\|_\infty$. We resolve both  issues in two steps: 
\begin{enumerate}
    \item  Sample and publish a topology (i.e., a new edge set) according to a predefined distribution.
    \item Re-weigh the edges in the new edge set.
\end{enumerate}

Step (a) can be regarded as sampling a topology according to the distribution induced by the exponential mechanism. Liu et al. \cite{liu2023optimal} shows that one can find a distribution of topologies that are computationally efficient to sample from while preserving both privacy and utility. In particular, they sample a new edge set $S$ of size $k$ ($k\geq |E|$) with probability proportional to 

\begin{equation}\label{e.distribution_intro}
  \forall S \subseteq [N] \text{ and } |S| = k, Pr[S] \propto \prod_{e\in S}\exp(\epsilon\cdot w_e).
\end{equation}
Compared to naively applying the exponential mechanism, sampling from this distribution removes the $\log (W)$ dependency on the maximum weight $W$. Liu et al. \cite{liu2023optimal} showed that it takes $O(mn^2)$ time to exactly sample from the distribution in \Cref{e.distribution_intro}. In this paper, we show that using approximate sampling (which results in $\delta \neq 0$), this can be done in linear time.

To get a linear time algorithm that samples from such distribution, we use the Markov Chain Monte Carlo (MCMC) method, specifically, the \textit{basis-exchange walk} on all edge sets with size $k$. We show that the distribution defined in \Cref{e.distribution_intro} is  {\em strongly log-concave}. Using standard modified log-Sobolev inequality for strongly log-concave distributions~\citep{cryan2019modified, anari2021log}, the basis-exchange walk is rapidly mixing, and we get an almost linear time approximate sampler. To be more precise, we introduce some useful definitions and facts that help us analyze the sampling subroutine. We consider the task of sampling a configuration from the two-spin system $\{0,1\}^N$ (namely a subset $S$ of $[N]$) according to a distribution $\pi$, where $N = {n\choose 2}$. One of the methods to characterize distribution $\pi$ is to use the generating polynomial, where 
$$\text{for any } z\in \mathbb{R}^N, \quad g_{\pi}(z) := \sum_{S\in [N]}\pi(S)\prod_{i\in S}z_i.$$
If the degree of the polynomial $g(\cdot):\mathbb{R}^N\rightarrow \mathbb{R}$ of each term is a constant $k$, then we say $g$ is $k$-homogeneous. If the degree of each variable $z_1,\cdots, z_N$ is at most $1$ in $g(\cdot)$, then we say $g$ is {\em multiaffine}. Clearly for any distribution $\pi$ on $\{0,1\}^N$, $g_\pi(\cdot)$ is multiaffine. An important property we are interested in is log-concavity. A polynomial $g(\cdot)$ with non-negative coefficients is log-concave if its logarithm is concave at $z\in \mathbb{R}_{\geq 0}^N$, i.e., the Hessian of $\log g$, $\nabla ^2\log g := (\partial_i \partial_j \log g)_{ij}$ is negative semi-definite for any $z\in \mathbb{R}_{\geq 0}^N$. Here, $\partial_i g$ is the partial derivative of $g$ with respect to its $i$-th coordinate. Another equivalent condition of log-concavity is that for any $z,y \in \mathbb{R}_{\geq 0}^N$ and $\lambda\in (0,1)$, 
$g(z + (1-\lambda)y) \geq g(z)^\lambda g(y)^{1-\lambda}.$

\noindent We need a stronger condition than log-concavity.  We say $g$ is \emph{strongly log-concave} if for any $I\in [N]$ and $I = \{i_1, \cdots, i_k\}$, $\partial_{i_1}\cdots\partial_{i_k} g$ is log-concave at $z\in \mathbb{R}^N_{\geq 0}$.

Given $k\in \mathbb{N}$ and $k\leq N$, basis-exchange walk is a commonly used Markov chain Monte Carlo method to sample one of the ${N\choose k}$ subsets of $[N]$ according to some given distribution $\pi$. We define the basis-exchange walk by the following procedure:

   \begin{itemize}
    \item Initialize $S_0  \subset [N]$ such that  $|S_0| = k$.
    \item At time $t = 1,2,\cdots,\infty$:
    \begin{enumerate}
      \item Choose an $e\in S_{t-1}$ uniformly at random, let $S_t' = S_{t-1}\backslash \{e\}$.
      \item Choose an element $y$ in $[N]\backslash S_t'$ with probability $\propto \pi(S_t'\cup \{y\})$, let $S_{t} = S_t'\cup \{y\}$. 
    \end{enumerate}
    \end{itemize}

It has been known that the basis-exchange walk enjoys rapid mixing if the target distribution is strongly log-concave (see also Lemma \ref{l.mixing}).

\begin{remark}
  By the analysis of privacy in Section \ref{sec:privacy} and Theorem \ref{t.utility}, one can find that our technique also gives an almost linear time algorithm for approximately simulating the exponential mechanism which samples from exponentially many candidates (each represented by a vector in $\mathbb{R}^{N}$ for some $N\in \mathbb{N}$) with $\ell_1$ norm as its scoring function under a $k$-sparsity ($k \leq N$) constraint. This results in an almost linear time $(\epsilon,\delta)$-differentially private algorithm for this task. We believe such auxiliary property of our technique is of independent interest.
\end{remark}

\paragraph{(2) High-pass filter} Although it has not yet been applied to private graph analysis, for sparse datasets, by adding Laplace noise to its histogram and filtering the output by setting a threshold, it is possible to obtain a linear additive error in terms of $\ell_1$ norm~\citep{cormode2012differentially}. If we apply such ideas to the graph settings, then there is a natural algorithm that outputs a synthetic graph while preserving differential privacy and takes $O(n^2)$ time: 

\begin{enumerate}
  \item Add independent Laplace noise from $\texttt{Lap}(1/\epsilon)$ on each pair of vertices. 
  \item Set $t =  \frac{c\log (n)}{\epsilon}$ for some large enough constant $c$. 
  \item For each $1\leq e \leq {n\choose 2}$, if the perturbed weight $\widehat{w}_e\leq t$, then set $\widehat{w}_e = 0$.
\end{enumerate}
We call this strategy the ``naive filtering algorithm''. By an elementary argument in Section \ref{sec:high_pass}, we see that this simple scheme achieves linear error on approximating cut size. 

Now if we set the threshold large enough, then with high probability, if the edge $e$ was not present in $G$, then it will be filtered out. This allows one to get an expected linear time algorithm. Formally, if the number of edges, $m$, is public, then one can first sample an integer $k$ from the binomial distribution $\texttt{Bin}(q,p)$, where $q = {n\choose 2} - m$ is the number of edges with zero weight and $p$ be the probability that $z\sim \texttt{Lap}(1/\epsilon)$ exceeds the threshold. 
Cormode et al. \cite{cormode2012differentially} applied this strategy to approximate the histogram of some given dataset, instead of approximating the size of $(S,T)$-cuts. 
However, since $k$ is a random variable, such a scheme only provides an algorithm that runs in expected linear time and it cannot guarantee to preserve the $m$-sparsity of the output graph with probability $1$, two major requirements in large-scale graph analysis. 

To address this issue, we propose a truncated version of the filtering algorithm which only relies on two simple modifications:
(a) Given $\epsilon,\delta$, change the threshold $t$ from $\frac{c\log (n)}{\epsilon}$ to $\frac{2c\log (n/\delta)}{\epsilon}$, and (b) project the naive filtering algorithm to the edge set $E$.

\noindent We show in \Cref{sec:high_pass_proofs} 
that this scheme provides a differentially private algorithm that always terminates in linear time while achieving linear error on approximating the size of all $(S, T)$-cuts.

\section{Private spectral and cut approximation}
\label{s.exp_linear}
This section presents our algorithms for differentially private cut and spectral approximation. For the sake of exposition, as in \cite{aumuller2021differentially, cormode2012differentially, eliavs2020differentially, gupta2010differentially}), we first assume the number of edges, $m=|E|$ is publicly known. However, in analyzing our theorems (Appendix \ref{app:proof_of_main}), we analyze where $|E|$ is confidential. Notably, these two assumptions have no fundamental barrier since we can add a Laplace noise on $|E|$ and run existing algorithms on the perturbed value, which only incurs a small error. 

\subsection{$\widetilde{O}(m)$ time algorithm using exchange walk}\label{sec:exchange_walk}

The idea of our first algorithms is based on Lemma \ref{lem:main}, which promises a linear time algorithm on sampling sparse candidates. In particular, we consider all possible typologies with $\widetilde{O}(m)$ edges as candidates. For any weighted graph $G = ([n],E,w)$, the algorithm can be divided into two stages:
  {(1)} Approximately sample a topology $\widehat{E}$, i.e., a subset of $[N]$ according to some predefined distribution with respect to $G$.   
  {(2)} Publish the weights of edges in $\widehat{E}$ by the Laplace mechanism.
In particular, we define $\pi$ be the distribution over the topology of graphs defined as following:
\begin{equation}\label{eq:distribution}
    \forall S \in 2^{[N]} \text{ and } |S| = k, \pi[S]\propto \prod_{e\in S} \exp(\epsilon\cdot w_e).
\end{equation}
We explain in Appendix \ref{sec:proof_main} this distribution is exactly the distribution defined in Lemma \ref{lem:main}, if we specify $x$ by the edge weights of the input graph. In stage 1, we sample an $S\subseteq [N]$ of size $|E|$ according to $\pi$. Intuitively, edges with larger weights should be more likely to be included. When $\epsilon = 0$, sampling from the distribution $\pi$ is equivalent to choosing an $S\in {[N]\choose |E|}$ uniformly at random, which satisfies perfect privacy.
After that, we let the chosen $|E|$ pair of vertices be the edge set $\widehat{E}$, and add Laplace noise on each edge in $\widehat{E}$ according to their weight in $G$. This results in a synthetic graph with exactly $|E|$ edges, which gives probability $1$ bound on the sparsity of the output graph.

We present the algorithm in Algorithm \ref{alg1_simple}.

\begin{algorithm}[h]
	\caption{{Private cut approximation by $T$ steps of basis-exchange walk}}\label{alg1_simple}
	\KwIn{A graph $G\in \mathbb{R}^N_{\geq 0}$, privacy budgets $\epsilon$, $\delta$.}
	\KwOut{A synthetic graph $\widehat{G}$.}
    Let $S_0 \leftarrow E$\;
    Set $T\leftarrow O(|E|\cdot (\epsilon + \log n + \log (1/\delta)))$\;
    \For{$t = \{1,2,\cdots, T\}$}{
      Choose an $e\in S_{t-1}$ uniformly at random and update $S_t' \leftarrow S_{t-1}\backslash \{e\}$\;
      Choose an element $x$ in $[N]\backslash S_t'$ with probability $\propto \pi(S_t'\cup \{x\})$, let $S_{t} \leftarrow S_t'\cup \{x\}$\;
    }
    Let $\widehat{E}\subset[N]$ be the edge set corresponding to $S_T$ \;
    \For{$e\in \widehat{E}$}{
        Draw an independent Laplace noise $Z \sim \lap(1/\epsilon)$ and set 
        $w_e  \leftarrow \max \{0, w_e + Z$\} \;
    }
    \For{$e'\in [N] \land e' \notin \widehat{E}$}{
        $w_{e'} \leftarrow 0$ \;
    }
    \Return{$\widehat{G} = (V,\widehat{E})$}.
\end{algorithm}

Our algorithm matches the non-private setting not only in running time for outputting the synthetic graph but also for post-processing. The condition $|\widehat{E}| = O(|E|)$ means that for any linear queries $q\in [0,1]^N$, we actually need only $O(|E|)$ time to compute the answer by $\widehat{G}$, instead of $O(n^2)$ time. This matches the running time in the non-private setting. The following theorem describes the privacy, utility, and resource usage of Algorithm \ref{alg1_simple}:

\begin{theorem}\label{t.main}
  Fix any $\epsilon>0$ , $\delta\in (0,1)$  and $\beta>0$.  There is an $(\epsilon,\delta)$-differentially private algorithm, such that on input an $n$ node $m$ vertices weighted graph $G = (V,E)$, it runs in $O(m\log(n/\delta)\log n)$ time, $O(m)$ space, and outputs a $\widehat{G} = (V,\widehat{E})$ with exactly $m$ edges such that, with probability at least $1-2\beta - \frac{\delta}{1+e^\epsilon}$, $Eval(G,\widehat{G})= {O}\left(\frac{{m\log (n/\beta)}}{\epsilon}\right).$ Further, we have $$\left\|L_G-L_{\widehat{G}}\right\|_2 = O\left(\frac{d_{\mathsf{max}}\cdot \log(n)}{\epsilon} + \frac{\log(1/\beta)\log^2n}{n\epsilon^2}\right).$$
\end{theorem}

\subsection{$\widetilde{O}(m)$-time algorithm using high-pass filter}\label{sec:high_pass}

In this section, we propose a new and simpler linear time algorithm while achieving a slightly worse performance (in terms of the dependency on $\delta$) compared to the algorithm in \Cref{sec:exchange_walk}. The idea is to use a large enough threshold to ``filter'' the edges with zero weights so that we do not have to add noise on all $O(n^2)$ pairs of vertices. A similar idea is applied in Cormode et al. \cite{cormode2012differentially} to publish a histogram of datasets; however, even if $|E|$ is public, the filtering algorithm in Cormode et al. \cite{cormode2012differentially} does not promise to keep $|E|$-sparsity, and the running time of their algorithm is linear time only in expectation. To resolve both issues, in Algorithm \ref{alg2}, we make two changes to the  filtering algorithm in \cite{cormode2012differentially}: for $\delta\in (0,1)$,  we use the threshold $\frac{c\log(2n/\delta)}{\epsilon}$ instead of $\frac{c\log(n)}{\epsilon}$ and run the filtering algorithm only on  $E$.

\begin{algorithm}[h]
	\caption{{Releasing a synthetic graph based on filtering algorithm}}\label{alg2}
	\KwIn{A graph $G\in \mathbb{R}^N_{\geq 0}$, privacy budgets $\epsilon$, $\delta$.}
	\KwOut{A synthetic graph $\widehat{G}$.}

    Set the threshold $t = \frac{2\log(2n/\delta)}{\epsilon}$\;
    \For{$e\in E$}{
      Draw an independent Laplace noise $Z \sim \lap(1/\epsilon)$  and set $\widehat{w}_e\leftarrow {w_e + Z}$ \;
      If $\widehat{w}_e \leq t$, then set $\widehat{w}_e \leftarrow 0$.
      
    }
    Let $\widehat{E} = \{e\in [N]~:~\widehat{w}_e >0\}$ be the new edge set \;

    \Return{$\widehat{G} = (V,\widehat{E})$}.

\end{algorithm}

Clearly, Algorithm \ref{alg2} runs in linear time and linear space, and it preserves exact $|E|$-sparsity with probability $1$. A less obvious fact is that Algorithm \ref{alg2} also preserves $(\epsilon,\delta)$-differential privacy, since we run the filtering algorithm only on the confidential edge set. Intuitively, if $G'$ is a neighboring graph of $G$ with an extra edge with a weight less than $1$, then this edge will be filtered with high probability (decided by $\delta$). Conditioned on that edge being filtered, the output distribution will be the same for both $G$ and $G'$. Theorem \ref{t.ana_on_alg2} proved in \Cref{sec:high_pass_proofs} formalizes this intuition.

\begin{theorem}\label{t.ana_on_alg2}
  For any $\epsilon>0$ and $\delta\in (0,1)$, Algorithm \ref{alg2} preserves $(\epsilon,\delta)$ differential privacy and, with probability at least $1-\delta$, outputs a $\widehat{G}$ such that 
  $Eval(G,\widehat{G}) =  O\left(\frac{|E|\log (n/\delta)}{\epsilon}\right).$
\end{theorem}
Algorithm \ref{alg2} also preserves the spectrum of the original graph, namely the quadratic form of $L_G$, where $L_G\in \mathbb{R}^{n\times n}$ is the Laplacian of $G$. It is noteworthy that the analysis on spectral approximation is not as straightforward as Algorithm \ref{alg2}'s appearance might suggest (see \Cref{sec:high_pass_proofs} for a proof).
\begin{theorem}\label{t.spectral}
  For any $G=([n],E)$ with bounded maximum unweighted degree $d_{\mathsf{max}}\leq n-1$, with high probability, Algorithm \ref{alg2} outputs a $\widehat{G} = (V,\widehat{E})$ such that for all $x\in \mathbb{R}^n$ with $\|x\|_2 \leq 1$, 
  $|x^\top L_Gx - x^\top L_{\widehat{G}}x| =  O\left(\frac{d_{\mathsf{max}}\log(n/\delta)}{\epsilon}\right).$
\end{theorem}

In Table \ref{t.2}, we compare our two linear additive error approaches on cut approximation when the number of edges is a publicly known parameter. The basis-exchange walk approach and the truncated filtering approach show different trade-offs with respect to $\delta$. As $\delta \to 0$, the accuracy of the truncated filtering algorithm gets worse, but the accuracy for the basis-exchange walk is unchanged. In contrast, the run-time of the basis-exchange walk gets worse when $\delta \rightarrow 0$, while the run-time of the truncated filtering algorithm is not affected by $\delta$. Indeed, neither algorithms dominate the other. In Section \ref{sec:simulation}, we also show that the algorithm based on the basis-exchange walk empirically achieves better accuracy in spectral apprixmation, albeit with a slightly larger time consumption.

\begin{table*}[t]
  \centering
  \caption{The comparison of our efficient and purely additive error approaches.}\label{t.2}

  \begin{tabular}{|c|c|c|c|c|}
      
      \hline
      \textbf{Method} &  \makecell[c]{\textbf{Privacy}} & \makecell[c]{\textbf{Accuracy}} &  \makecell[c]{\textbf{Run-time}} & \makecell[c]{\textbf{Space}}   \\
      \hline 
      \makecell[c]{Basis-exchange walk (Thm. \ref{t.main})} & $(\epsilon,\delta)$-DP & ${O}\left(\frac{m\log n }{\epsilon}\right)$ & {${O}\left(m\log ({n\over \delta})\right)$}  & {$2m$} \\
      \hline 
      \makecell[c]{Truncated Filtering algorithm (Thm. \ref{t.ana_on_alg2})} & $(\epsilon,\delta)$-DP & ${O}\left(\frac{m\log (n/\delta) }{\epsilon}\right)$ & {${O}\left(m\right)$} & {$m$} \\
      \hline
  \end{tabular}
\end{table*}

\section{Empirical Simulations}\label{sec:simulation}

In this section, we implement a synthetic experiment to demonstrate the advantages of our algorithms in terms of both efficiency and accuracy. In particular, we compare our algorithms with $3$ major previous algorithms for outputting a synthetic graph for cut or spectral approximation: Johnson-Lindenstrauss Transformation~\cite{blocki2012johnson}, Analyze Gauss~\cite{dwork2014analyze} and Liu et al.~\cite{liu2023optimal}. The reason we do not compare with private mirror descent~\cite{eliavs2020differentially} is because that \cite{eliavs2020differentially} does not provide any theoretical guarantee on graph spectral approximation, and evaluating the maximum error over all $2^n$ cuts would be intangible.

We set privacy budget as $\epsilon=1,\delta = n^{-10}$. For JL mechanism, we set the multiplicative factor $\eta$ be $0.1$. The experiment is conducted on MacBook Air M1. All test results are the average of 5 runs.
\begin{table*}[b]
  \centering
  \caption{The comparison in run-time.}\label{tab:runtime}

  \begin{tabular}{|c|c|c|c|c|c|}
      
      \hline
      \textbf{$n$} &  \makecell[c]{\textbf{Algorithm \ref{alg1_simple}}} & \makecell[c]{\textbf{Algorithm \ref{alg2}}} &  \makecell[c]{\textbf{Analyze Gauss}\\ \cite{dwork2014analyze}} & \makecell[c]{\textbf{JL mechanism}\\ \cite{blocki2012johnson}} & \makecell[c]{\textbf{Liu et al.}\\ \citet{liu2023optimal}}  \\
      \hline 
      \makecell[c]{$100$} & 0.043 s & 0.011 s & 0.028 s  & 1.584 s & 2.376 s\\
      \hline 
    \makecell[c]{$10^3$} & 0.362 s & 0.093 s & 2.722 s  & $>$ 300 s & $>$ 300 s\\
      \hline 
     \makecell[c]{$10^4$} & 2.103 s & 0.983 s & 284.263 s  & $>$ 300 s & $>$ 300 s\\
      \hline 
     \makecell[c]{$10^5$} & 22.742 s & 10.029 s & $>$ 300 s  & $>$ 300 s & $>$ 300 s\\
      \hline 

  \end{tabular}
  \label{table:spectralrun}
\end{table*}

\textbf{Run-time.}
We first evaluate the run-time for outputting a synthetic graph. For the dataset, we use Erdős Rényi random graph model $G:=G(n,c/n)$ (for constant $c$) to generate a synthetic graph. At every different scale of the input, we always set the expected average degree be constant to make sure that with high probability, the graph is sparse. We test the run time of different algorithms on 5 different scales (from $10^2$ to $10^5$).
From \Cref{table:spectralrun}, we see that: (i) all previous algorithms become rather slow even on moderately large graphs (with $10^4$ vertices) and (ii) on different scales of the input, the run-times of our algorithms exhibit an obvious linear growth trend, which verifies our theory on time complexity.

\textbf{Accuracy.}
We show the error only for spectral approximation as estimating worst-case cut approximation is not feasible. To better study the accuracy of algorithms on weighted graphs, we use three different scales of edge weights: $W = 1$ (which is also the unweighted case), $W = \sqrt{n}$ and $W = n$. For the sake of time-consuming, we set $n = 100$. The results on spectral error $\|L_G - L_{\widehat{G}}\|_2$ are listed in \Cref{table:spectral}. From the table, we see that with the average weight value increasing, the spectral error of our algorithms (as well as Analyze Gauss) shows virtually no change. Notably, our algorithms consistently focus its spectral error on the average degree even for graphs with significantly high weights, thus confirming our Theorem \ref{t.spectral}. The superiority of our algorithm over Analyze Gauss might be from its novel approach to incorporating noise. Our Algorithm \ref{alg1_simple} achieves better accuracy, our conjecture is that compared to Algorithm \ref{alg2}, it may have a smaller constant or that the error of Algorithm \ref{alg1_simple} does not depend on $\delta$ (see also the discussions in Section \ref{sec:high_pass}). Also, even though the JL mechanism is supposed to have asymptotically the same additive error on spectral approximation, the experiment shows that the JL mechanism produces a rather large error compared to other mechanisms. This is because the JL mechanism includes a multiplicative approximation, which worsens as the weight increases. Further, the constant in the error bound of the JL mechanism is large.

\begin{table*}[t]
  \centering
  \caption{The comparison in accuracy of spectral approximation (I).}\label{tab:accuracy1}

  \begin{tabular}{|c|c|c|c|c|c|}
      \hline
      \textbf{$n$} &  \makecell[c]{\textbf{Algorithm \ref{alg1_simple}}} & \makecell[c]{\textbf{Algorithm \ref{alg2}}} &  \makecell[c]{\textbf{Analyze Gauss}\\ \cite{dwork2014analyze}} & \makecell[c]{\textbf{JL mechanism}\\ \cite{blocki2012johnson}} & \makecell[c]{\textbf{Liu et al.}\\ \citet{liu2023optimal}}  \\
      \hline 
      \makecell[c]{$1$} & 23.935 & 32.638 & 43.973 & 817.930 & 19.249\\
      \hline 
    \makecell[c]{$\sqrt{n}$} & 30.127 & 31.931 & 45.812 & 1861.872 & 29.587\\
      \hline 
     \makecell[c]{$n$} & 25.347 & 28.189 & 41.723 & 2769.758 & 25.328\\
      \hline 

  \end{tabular}
  \label{table:spectral}
\end{table*}

To understand the growth rate of the spectral error for different sizes, we choose to test our algorithms and Analyze Gauss on unweighted graphs within the input range from $n = 200$ to $n = 1000$. The results are listed in \Cref{table:spectralaccuracy}. That is, the error of our algorithm does not increase (and perform better) if the maximum unweighted degree is a constant as theory predicts.

\begin{table*}[h]
  \centering
  \caption{The comparison in accuracy of spectral approximation (II).}\label{tab:accuracy2}

  \begin{tabular}{|c|c|c|c|}
      
      \hline
      \textbf{$n$} &  \makecell[c]{\textbf{Algorithm \ref{alg1_simple}}} & \makecell[c]{\textbf{Algorithm \ref{alg2}}} &  \makecell[c]{\textbf{Analyze Gauss}\cite{dwork2014analyze}}  \\
      \hline 
      \makecell[c]{$200$} & 24.413 & 34.589 & 59.559 \\
      \hline 
            \makecell[c]{$400$} & 24.466 & 38.262 & 87.396 \\
      \hline 
            \makecell[c]{$600$} &24.874 & 36.599 & 130.889 \\
      \hline 
            \makecell[c]{$800$} & 25.097 & 38.250 & 161.522 \\
      \hline 
            \makecell[c]{$1000$} & 25.875 & 38.395 & 178.136 \\
      \hline 

  \end{tabular}
  \label{table:spectralaccuracy}
\end{table*}

\section{Conclusion and Limitations}
In this paper, we propose private algorithms on additive cut approximation and spectral approximation that run in linear time and linear space, such that the compute time and space required by both the curator and the analyst is almost the same as in the non-private setting. One of our methods can also be applied to efficiently sample from the exponential mechanism defined over candidates with the sparsity constraint, which is a non-convex set. Given that resource matters in large-scale real industrial deployments and that real-world graphs are sparse, we believe our results will positively impact private graph analysis in real-world industrial deployments. 
While the setting studied in this paper captures most of the practical query sets and naturally occurring graphs in industry, the utility of our algorithms on cut approximation is sub-optimal. it remains an open question whether there is an efficient (or even linear time) algorithm that, with probability at least $1-o(n^{-c})$, answers all $(S,T)$-cut queries on a graph with the same $\widetilde {O}(\sqrt{n|E|})$ error rate as in Liu et al.~\citet{liu2023optimal}.

 \medskip
 \subsubsection*{Acknowledgement} JU research is supported by Rutgers Decanal Grant no.
302918.
\clearpage
\bibliographystyle{alpha}
\bibliography{privacy}

\newpage
\appendix

\section{Technical background}

\subsection{Differential privacy}
Differential privacy, proposed by \cite{dwork2006calibrating}, is a widely accepted notion of privacy. Here, we formally define differential privacy. Let $\mathcal{Y}$ be the data universe where $|\mathcal{Y}| = d$. Every dataset $D\in \mathcal{Y}^n$ with cardinality $n$ can be equivalently written as a histogram $D\in \mathbb{N}^{d}$, where each coordinate of $D$ records the number of occurrences of the corresponding element in $\mathcal{Y}$. Here, we write $\mathcal{X}\subset \mathbb{N}^d$ as the collection of datasets we are interested in. Also for $i\in [d]$, we use $D[i]$ to denote the value in $i$-th coordinate of $D$.

\begin{definition}
    [Neighboring dataset] For any two datasets $x,y\in \mathcal{X}$, we say that $x$ and $y$ are neighboring datasets if and only if $x$ and $y$ differ in at most one coordinate by $1$. Namely, there is an $i\in [d]$ such that $|x[i] - y[i]|\leq 1$, and $x[j] = y[j]$ for all $j\in [d]\backslash \{i\}$.
\end{definition}
\noindent If $x$ and $y$ are neighboring datasets, we also write $x\sim y$. The differential privacy measures privacy loss by distributional stability between neighboring datasets:
\begin{definition}
    [$(\epsilon,\delta)$-differential privacy~\citep{dwork2006our}] Let $\mathcal{A}:\mathcal{X}\rightarrow \mathcal{R}$ be a (randomized) algorithm, where $\mathcal{R}$ is the output domain. For fixed $\epsilon >0$ and $\delta\in [0,1)$, we say that $\mathcal{A}$ preserves $(\epsilon,\delta)$-differential privacy if for any measurable set $S\subset \mathcal{R}$ and any $x,y\in \mathcal{X}$ be a pair of neighboring datasets, it holds that 
    $$Pr[\mathcal{A}(x)\in S] \leq Pr[\mathcal{A}(y)\in S]\cdot e^{\epsilon} + \delta.$$
\end{definition}

\noindent Here is a useful definition related to $(\epsilon,\delta)$-DP:
\begin{definition}
  ($(\epsilon,\delta)$-indistinguishable) Let $\Omega$ be a ground set and $\mu_1, \mu_2$ be two distributions with support $\Omega_1\subset \Omega$, $\Omega_2\subset \Omega$ respectively. We say that $\mu_1$ and $\mu_2$ are $(\epsilon,\delta)$-indistinguishable for $\epsilon>0$ and $\delta \in (0,1)$ if for any $S\subset \Omega$, it holds that 
  $$Pr_{\mu_1} (S) \leq Pr_{\mu_2}(S)\cdot e^\epsilon + \delta$$
  and 
  $$Pr_{\mu_2} (S) \leq Pr_{\mu_1}(S)\cdot e^\epsilon + \delta.$$
\end{definition}

A key feature of differential privacy algorithms is that they preserve privacy under post-processing. That is to say, without any auxiliary information about the dataset, any analyst cannot compute a function that makes the output less private. 
\begin{lemma}
    [Post processing~\citep{dwork2014algorithmic}] Let $\mathcal{A}:\mathcal{X}\rightarrow \mathcal{R}$ be a $(\epsilon,\delta)$-differentially private algorithm. Let $f:\mathcal{R}\rightarrow \mathcal{R}'$ be any function, then $f\circ \mathcal{A}$ is also $(\epsilon,\delta)$-differentially private.
\end{lemma}

Sometimes we need to repeatedly use differentially private mechanisms on the same dataset, and obtain a series of outputs.
\begin{lemma}[Basic composition lemma, \citep{dwork2006calibrating}]\label{l.composition}
     let $D$ be a dataset in $\mathcal{X}$ and $\mathcal{A}_1, \mathcal{A}_2,\cdots, \mathcal{A}_k$ be $k$ algorithms where $\mathcal{A}_i$ (for $i\in [k]$) preserves $(\epsilon_i,\delta_i)$ differential privacy, then the composed algorithm $\mathcal{A}(D) = (\mathcal{A}_1(D), \cdots, \mathcal{A}_2(D))$ preserves $(\sum_{i\in [k]}{\epsilon_i}, \sum_{i\in [k]}{\delta_i})$-differential privacy.
\end{lemma}

\begin{lemma}
    [Advanced composition lemma, \citep{dwork2010boosting}] 
    \label{l.adv_composition}
    For parameters $\epsilon>0$ and $\delta,\delta'\in [0,1]$, the composition of $k$ $(\epsilon,\delta)$ differentially private algorithms is a $(\epsilon', k\delta+\delta')$ differentially private algorithm, where 
    $$\epsilon' = \sqrt{2k\log(1/\delta')} \cdot \epsilon + k\epsilon (e^\epsilon - 1).$$
\end{lemma}

\noindent Now, we introduce some basic mechanisms that preserve differential privacy. First, we define the sensitivity of query functions.

\begin{definition}
    [$\ell_p$-sensitivity] Let $f:\mathcal{X}\rightarrow \mathbb{R}^k$ be a query function on datasets. The sensitivity of $f$ (with respect to $\mathcal{X}$) is 
    \[\mathsf{sens}_p (f) = \max_{x,y\in \mathcal{X},\atop x\sim y} \|f(x) - f(y)\|_p.\]
\end{definition}

The Laplace mechanism is one of the most basic mechanisms to preserve differential privacy for numeric queries, which calibrates a random noise from the Laplace distribution (or double exponential distribution) according to the $\ell_1$ sensitivity of the function. 

\begin{definition}
    [Laplace distribution] Given parameter $b$, Laplace distribution (with scale $b$) is the distribution with probability density function 
    $$\texttt{Lap}(x|b) = \frac{1}{2b} \exp\left(-\frac{|x|}{b}\right).$$
    Sometimes we use $\texttt{Lap}(b)$ to denote the Laplace distribution with scale $b$.
\end{definition}
\begin{lemma}[Laplace mechanism]\label{l.laplace}
     Suppose $f:\mathcal{X}\rightarrow \mathbb{R}^k$ is a query function with $\ell_1$ sensitivity $\mathsf{sens}_1(f)\leq \mathsf{sens}$. Then the mechanism
    $$\mathcal{M}(D) = f(D) + (Z_1,\cdots,Z_k)^\top$$
    is $(\epsilon,0)$-differentially private, where $Z_1,\cdots, Z_k$ are i.i.d random variables drawn from $\texttt{Lap}(\mathsf{sens}/\epsilon)$.
\end{lemma}

In many cases the output domain of the query function is discrete. For example, according to a private dataset, we would like to output a candidate with the highest score. In those cases, the Laplace mechanism is not applicable. A more fundamental mechanism for choosing a candidate is the {\em exponential mechanism}. Let $\mathcal{R}$ be the set of all possible candidates, and let $s:\mathcal{X}\times \mathcal{R}\rightarrow \mathbb{R}_{\geq 0}$ be a scoring function. We define the sensitivity of $s$ be $$\mathsf{sens}_s = \max_{x,y\in \mathcal{X}\atop x\sim y} |s(x) - s(y)|.$$ 

Now if we want to output a candidate that minimizes the scoring function, we define the exponential mechanism $\mathcal{E}$ for input dataset $D\in \mathcal{X}$ to be
$\mathcal{E}(D) := \text{Choose a candidate } r\in \mathcal{R}$ { with probability proportional to } $\exp\left(- \frac{{\epsilon \cdot s(D,r)}}{2\mathsf{sens}_s}\right)$ { for all } $r\in \mathcal{R}.$ 
We have the following lemmas for the privacy guarantee and utility of $\mathcal{E}(\cdot)$.
\begin{lemma}\label{l.exp}
    The exponential mechanism $\mathcal{E}(\cdot)$ sampling in some discrete space $\mathcal{R}$ is $\epsilon$-differentially private.
\end{lemma}

\begin{lemma}\label{l.exp_utility}
    Let $\mathcal{E}(\cdot)$ be an exponential mechanism sampling in some discrete space $\mathcal{R}$. For any input dataset $D$ and $t>0$, let $OPT = \min_{r\in \mathcal{R}}s(D,r)$, then we have that
    $$Pr\left[s(D,\mathcal{E}(D)) \geq OPT + \frac{2\mathsf{sens}_s}{\epsilon} (\log (|\mathcal{R}|) + t)\right] \leq \exp(-t).$$
\end{lemma}

The exponential mechanism $\mathcal{E}(\cdot)$ can also be generalized to sample from in a continuous space $\mathcal{R}$. 
Then Lemma \ref{l.exp} and Lemma \ref{l.exp_utility} can be naturally generalized to the following two lemmas:

\begin{lemma}\label{ec.pri}
    The exponential mechanism $\mathcal{E}(\cdot)$ sampling in some continuous space $\mathcal{R}$ is $\epsilon$-differentially private.
\end{lemma}
Fix a $D\in \mathcal{X}$, for any $t\geq 0$, we define $\mathcal{R}_t  = \{r\in \mathcal{R}, s(D,r) \leq OPT + t\}\subset \mathcal{R}$, and for a uniform distribution $\pi$ over $\mathcal{R}$, we write $\pi(\mathcal{R}_t)$ as the probability volume of $\mathcal{R}_t$.
\begin{lemma}\label{ec.acc}
    Let $\mathcal{E}(\cdot)$ be an exponential mechanism sampling in some continuous space $\mathcal{R}$, and $\pi$ be the uniform distribution on $\mathcal{R}$. For any input dataset $D$, let $OPT = \min_{r\in \mathcal{R}}s(D,r)$, then we have that
    $$Pr\left[s(D,\mathcal{E}(D)) \geq OPT + \frac{4\mathsf{sens}_s}{\epsilon} \left(\log \left(\frac{1}{\pi(\mathcal{R}_t)}\right) + t\right)\right] \leq \exp(-t).$$
\end{lemma}

\subsection{Spectral graph theory}

In this paper, the main application of our algorithm is to output a synthetic graph with $n$ vertices that approximately answers all cut queries. A weighted graph $G = (V,E)$ with non-negative edges can be equivalently written as a vector $G\in \mathbb{R}_{\geq 0}^N$, where we write $N = {n \choose 2}$, and let $G[i]$ ($i\in [N]$) be the edge weight in $i$-th pair of vertices. If some pair of vertices have no edge, the weight is $0$. Let $A_G\in \mathbb{R}^n\times \mathbb{R}^n$ be the symmetric adjacency matrix of $G$. That is,
$A_G[u,v] = A_G[v,u] = w_{uv}$, where $u,v\in V$ are vertices, and $w_{u,v}$ is the weight of edge $uv$. (If $u$ is not adjacent with $v$, then $w_{uv} = 0$.) For vertex $v\in V$, denote by $d_v$ the weighted degree of $v$, namely $d_v = \sum_{u\neq v} w_{vu}$. Let $D_G\in \mathbb{R}^n\times \mathbb{R}^n$ be a diagonal matrix where $D_G[i,i]$ is the weighted degree of node $i\in[n]$. Here, we define the edge adjacency matrix:

\begin{definition}
    [Edge adjacency matrix] Let $G$ be an undirected graph of $n$ vertices and $m$ edges. Consider an arbitrary orientation of edges, then $E_G\in \mathbb{R}^{m\times n}$ is the edge adjacency matrix where 
    $$E_G[e,v] = \left\{
        \begin{aligned}
            &+\sqrt{w_e}, &\text{if } v \text{ is } e\text{'s head,} \\
            &-\sqrt{w_e}, &\text{if } v \text{ is } e\text{'s tail,} \\
            &0, &\text{otherwize.}
        \end{aligned}
    \right.$$ 
\end{definition}
\noindent An important object of interest in graph theory is the Laplacian of a graph:
\begin{definition}
    [Laplacian matrix] For an undirected graph $G$ with $n$ vertices and $m$ edges, the Laplacian matrix $L_G\in \mathbb{R}^{n\times n}$ of $G$ is $$L_G = E_G^\top E_G.$$
\end{definition}

\noindent Equivalently, one can verify that $L_G = D_G - A_G$. Also, we note that for any graph $G$, $L_G$ is a positive semi-definite (PSD) matrix and $L_G \bm{1} = 0$, where $\bm{1}\in \mathbb{R}^n$ is the all one vector. Let $0 = \lambda_1(G)\leq \lambda_2(G)\leq\cdots\leq \lambda_n(G)$ be the non-negative eigenvalues of $L_G$. The following lemma illustrates the monotonicity of eigenvalues of Laplacian:

\noindent For any vector $x$ in $\mathbb{R}^n$, the quadratic form of $L_G$ is $x^\top L_G x  \geq 0$. In particular, one can verify that 
$$x^\top L_G x = \sum_{(u,v)\in E} w_{u,v} (x(u) - x(v))^2.$$
Our main task is to preserve the cut size. For any $S,T\subset V$, we simply write $\Phi_G(S)$ be the size of $(S,V\backslash {S})$-cut for a graph $G$, and $\Phi_G(S,T)$ be the size of $(S,T)$-cut. For any vertex $v\in V$, we use $\mathbf{1}_v\in \{0,1\}^n$ to denote the column vector with $1$ in $v$-th coordinate and $0$ anywhere else. Also, for a subset of vertices $S$ we use $\mathbf{1}_S\in \{0,1\}^n$ to denote the identity vector of $S$. It is easy to verify that 
$$\Phi_G(S) = \mathbf{1}_S^\top L_G  \mathbf{1}_S.$$
Therefore, if we obtain an approximation for the quadratic form of the Laplacian matrix $L_G$, it is equivalent to obtaining an approximation for the $(S,V\backslash {S})$-cut size.

We also need the following lemma to transform cut approximation bounds to a spectral one.
\begin{lemma}
  [Bilu and Linial \cite{bilu2006lifts}] 
  \label{l.spectral}
  Let $A\in \mathbb{R}^{n\times n}$ be a symmetric matrix such that $\|A\|_{\infty} \leq \ell$ and all diagonal entries of $A$ has absolute value less than $O(\alpha (\log(\ell/\alpha)+1))$. If for any non-zero $u,v\in \{0,1\}^n$ with $u^\top v = 0$, it holds that 
  $$\frac{|u^\top A v|}{\|u\|_2\|v\|_2} \leq \alpha,$$
  then  the spectral radius of $A$ is  $\|A\|_2 = O(\alpha (\log(\ell/\alpha)+1))$.
\end{lemma}

\subsection{Basics in Markov Chain Monte Carlo}

Let $\Omega$ be a finite state space. We write $\{X_t\}_{t\geq 0}$ be a Markov chain on $\Omega$ whose transition is determined by the transition matrix $P\in \mathbb{R}^{\Omega\times \Omega}$. We also use $P$ to denote a Markov chain if it doesn't arise any ambiguity. The Markov chain on $\Omega$ is \textit{irreducible} if for all $X,Y\in \Omega$, there exists a $t$ such that $P^t(X,Y)>0$, and it is \textit{aperiodic} if for all $X\in \Omega$, $\text{gcd}\{t>0|P^t(X,X)>0\} = 1$.
If a Markov chain is both irreducible and aperiodic, then it has a unique stationary distribution $\pi$ on $\Omega$ such that $\pi = P\pi$. If a distribution $\pi$ satisfies that 
$$\forall X,Y\in \Omega: \pi(X)P(X,Y) = \pi(Y)P(Y,X)$$
then we say the Markov chain $P$ is reversible with respect to $\pi$, which also implies that $\pi$ is the stationary distribution of $P$.

We use the \textit{mixing time} to measure the convergence rate of a Markov chain. For two distributions $\mu,\nu$ on $\Omega$, let 
$$\|\mu,\nu\|_{TV}  = \frac{1}{2}\|\mu-\nu\|_{1} = \max_{S\subset \Omega}|\mu(S) - \nu(S)|$$ denote the total variation distance. 
Then, the mixing time of $P$ with stationary distribution $\pi$ is defined by 
$$\forall 0<\epsilon <1:  T_{\text{mix}}(\epsilon) = \max_{x\in \Omega}\min\left\{t~|~d_{TV}(\pi,P^t(x,\cdot))\right\}.$$
Here, $P^t(x,\cdot)$ is the distribution of $X_t$ in $\{X_t\}_{t\geq 0}$ starting with $X_0 = x$.

In this paper, given a positive integer, $N\in \mathbb{N}_+$, we are interested in sampling a subset of $[N]$ with some specific size. For some $k\leq N$, we use ${[N] \choose k}$ to denote the collection of $S\subset[N]$ and $|S| = k$. A widely used Markov chain for sampling a subset of $[N]$ is \textit{basis-exchange walk} (\cite{feder1992balanced, anari2019log, cryan2019modified, anari2021log}), which has been formally defined in Appendix \ref{app:tech_overview}. It can actually be used to accomplish a more general task: sampling independent sets in a matroid. One can verify that the basis-exchange walk is a special case of {\em Glauber dynamics}  (\cite{martinelli1999lectures}), and thus it is irreducible and aperiodic. Also, its unique stationary distribution is $\pi$. Note that $\pi$ is a distribution on $2^{[N]}$. We use the generating polynomial $g_{\pi}(\cdot)$ to characterize distribution $\pi$, where

$$g_{\pi}(z):=\sum_{S\in [N]} \pi(S)\prod_{i\in S}z_i.$$ 
We are interested in the strong log-concavity of $g(\cdot)$. A non-negative function $g : \mathbb R^n \to \mathbb R$ is {\em log-concave} if its domain is a convex set, and if it satisfies the inequality
\[
g(tx+(1-t)y) \geq (g(x))^t (g(y))^{1-t}
\]
for all $x,y$ in the domain of $g$. 
Now, strong log-concavity is defined as follows:

\begin{definition}[Strong log-concavity]
\label{defn:stronglogconcavity}
Let $g:\mathbb{R}^N\rightarrow \mathbb{R}$ be a polynomial. We say $g$ is strongly log-concave if for any $I\in [N]$, $\partial_I g$ is log-concave at $z\in \mathbb{R}^N_{\geq 0}$. Here, $\partial_I g:= \partial_{i_1}\cdots\partial_{i_k} g$ for $I = \{i_1, \cdots, i_k\}$.
\end{definition}

Let $P\in \mathbb{R}^{2^N\times 2^N}$ be the Markov chain that transits subsets in $[N]$. It has been known that distributions with strong log-concavity satisfy good properties for mixing time:

\begin{lemma}[\cite{cryan2019modified, anari2021log}]\label{l.mixing}
Let $\pi$ be a $k$-homogeneous strongly log-concave distribution on ${[N]\choose k}$, and $P$ be the corresponding basis-exchange walk, then for any $\alpha>0$ and $S\in {[N]\choose k}$, 
$$\|P^t(S,\cdot) , \pi \|_{\text{TV}} \leq \alpha$$ 
as long as 
$$t \geq {r\cdot (\log \log (\pi(S)^{-1}) + 2\log (\alpha^{-1}) + \log 4)}.$$
\end{lemma}
\noindent Thus, we see that the mixing time of basis-exchange walk on subsets of size $r$ can be bounded roughly by $O(r(\log \log \pi(S)^{-1}) + \log (\alpha^{-1}))$, where $S$ is the initial state.

\section{Proof of Lemma \ref{lem:main}}\label{sec:proof_main}

We first give a formal restatement of Lemma \ref{lem:main} here:

\begin{lemma}[Restatement of Lemma \ref{lem:main}]\label{lem:main_formal}
    Fix $k,N,\epsilon$ such that $\epsilon \geq 0$, $k,N\in \mathbb{N}$ and $k\leq N$. Given a vector $x \in \mathbb R^N$ in the positive quadrant such that $\|x\|_0 \leq k$, let $\pi$ be the distribution $$ \pi [S] \propto e^{-\epsilon\|x-x|S\|_1}$$ with support $\{S\in \{0,1\}^N : \|S\|_0 = k \}$.
    Then, for any $0<\delta<1$, there exists an $\widetilde O_{\delta}(k)$ time random algorithm $\mathcal{A}$ that outputs an element in the same support such that if $\pi'$ is the output distribution of $\mathcal{A}$, then $\|\pi - \pi'\|_{TV} \leq \frac{\delta}{e^{2\epsilon} + 1}$.
\end{lemma}
Here, $x|S$ is the restriction of $x$ on $S$. That is, for any $ i \in [N]$, $(x|S)_i = x_i$ if $S_i = 1$ and $(x|S)_i = 0$ otherwise. In particular, if $S^*$ is the support of $x$, then $x = x|S^*$ and thus $S^*$ maximizes the probability of been chosen. We first show how the distribution defined in Lemma \ref{lem:main_formal} is connected to the distribution in \Cref{eq:distribution}.

For any fixed $x\in \mathbb{R}_{\geq 0}^N$, we have that
\begin{equation*}
    \begin{aligned}
        \Pr[S] &\propto \exp(-\epsilon\|x-x|S\|_1) = \exp\left( -\epsilon \sum_{i\in [N]} x_i - (x|S)_i\right)\\
        &= \exp\left({ -\epsilon \sum_{i\notin S} x_i}\right) \propto \exp \left({\epsilon \sum_{i\in [N]}x_i}\right)\cdot  \exp\left({ -\epsilon \sum_{i\notin S} x_i}\right)\\
        & = \exp \left(\epsilon \sum_{i\in S} x_i \right) = \prod_{i\in S} \exp(\epsilon\cdot x_i).
    \end{aligned}
\end{equation*}

Here, we write $i\in S$ if $S_i = 1$ otherwise $i\notin S$. Thus, if we let $N = {n\choose 2}$ and $x$ be the weights of the input graph, then sampling from the distribution in Lemma \ref{lem:main_formal} is equivalent to sampling from the distribution in \Cref{eq:distribution}. To sample from the distribution 
\begin{equation}\label{eq:distribution_restate}
    \forall S \in 2^{[N]} \text{ and } |S| = k, \pi[S]\propto \prod_{e\in S} \exp(\epsilon\cdot x_e),
\end{equation}

we apply the basis-exchange walk which is commonly used in Markov chain Monte Carlo method. The algorithm is summarized in Algorithm \ref{alg:exchange_walk}.

\begin{algorithm}[h]
	\caption{{Basis-exchange walk}}\label{alg:exchange_walk}
	\KwIn{A parameter $k$, time limit $T$}
	\KwOut{A subset $S\subset [N]$.}

    Initialize $S_0\subseteq N$ such that $|S_0| = k$\;

    \For{$t = \{1,2,\cdots, T\}$}{
      Choose an $e\in S_{t-1}$ uniformly at random, let $S_t' \leftarrow S_{t-1}\backslash \{e\}$\;
      Choose an element $y$ in $[N]\backslash S_t'$ with probability $\propto \pi(S_t'\cup \{y\})$, let $S_{t} \leftarrow S_t'\cup \{y\}$\;
    }
    \Return{$S_T$}.
\end{algorithm}

To analyze how the distribution of $S_T$ approximates $\pi$ in \Cref{eq:distribution_restate}, we first show that the target distribution $\pi$ satisfies strong log-concavity.

\begin{lemma}\label{l.f_is_concave}
 Fix any $x\in \mathbb{R}_{\geq 0}^N$, $k\geq 0, \epsilon>0$, the generating polynomial of distribution $\pi$ defined in  \Cref{eq:distribution_restate} is strongly log-concave at any $z\in \mathbb{R}_{\geq 0}^N$.
\end{lemma}
\begin{proof}
  We use the celebrated log-concavity of the basis generating polynomial of any matroids (see, e.g., \cite{anari2018log,adiprasito2018hodge}). They imply that, for any $k, N$ such that $k\leq N$, if $\mathcal{B}\subseteq {[N]\choose k}$ is a collection of subsets of size $k$, then the polynomial 
  \begin{equation}\label{e.matroid}
      \widehat{g}(z_1,\cdots, z_N) :=\sum_{S\in \mathcal{B}}\prod_{e\in S}z_e
  \end{equation}
  is log-concave. Let $g(\cdot):\mathbb{R}^N_{\geq 0}\rightarrow \mathbb{R}$ be the generating polynomial of distribution $\pi$ defined in  \Cref{e.distribution}. Let $Z(G,\epsilon)$ be the partition function of $\pi$ with respect to $G$ and $\epsilon$, and let $f(e) = \exp(\epsilon\cdot w_e)$ for any $e\in [N]$.  That is,
  $Z(G,\epsilon) = \sum_{S\in {N\choose k}} \prod_{e\in S}f(e)$.  
  We use $M:\mathbb{R}^N\times \mathbb{R}^N$ to denote a diagonal matrix where $M_{ee} = f(e)$ for $e\in [N]$. Then, we have that for any $z\in \mathbb{R}_{\geq 0}^N$, 
  \begin{equation*}
    \begin{aligned}
      g(z) &= \sum_{S\in {N\choose k}} \pi(S) \prod_{e\in S} z_e = \frac{\sum_{S\in {N\choose k}} \prod_{e\in S}f(e) \prod_{e\in S} z_e}{Z(G,\epsilon)} \\
      &= \frac{1}{Z(G,\epsilon)} \sum_{S\in {N\choose k}}\prod_{e\in S} f(e)z_e 
    \end{aligned}
  \end{equation*}
  Thus, for any $z,y\in \mathbb{R}_{\geq 0}^N$ and $\lambda\in (0,1)$, let $\mathcal{B} = {[N] \choose k}$, then for $\bar z=\lambda z + (1-\lambda)y$, 
  \begin{equation*}
    \begin{aligned}
      g(\bar z) &= \frac{1}{Z(G,\epsilon)}\sum_{S\in {N\choose k}}\prod_{e\in S} f(e)(\lambda z_e + (1-\lambda)y_e)\\
      &=\frac{1}{Z(G,\epsilon)} \cdot \widehat{g}(\lambda Mz + (1-\lambda)My)\\
      &\geq\left( \frac{1}{Z(G,\epsilon)} \cdot \widehat{g}(Mz)\right)^\lambda  \left(\frac{1}{Z(G,\epsilon)} \cdot \widehat{g}(My)\right)^{1-\lambda}\\
      & = g^\lambda(z)\cdot g^{1-\lambda}(y).
    \end{aligned}
  \end{equation*}
  Here, $\widehat{g}(\cdot)$ is defined in \Cref{e.matroid}. Thus, $g(\cdot)$ is log-concave. Further, it's easy to verify that for any $i\in[N]$, $\partial_i g$ is the generating polynomial of distribution on ${[N]\backslash \{e\} \choose k-1}$ and is $(k-1)$-homogeneous. By the same argument, $\partial_i g$ is also log-concave. By induction, we have that for any $I\in [N]$, $\partial_I g$ is log-concave at $z\in \mathbb{R}^N_{\geq 0}$. Thus, $g(\cdot)$ is strongly log-concave at $z\in \mathbb{R}^N_{\geq 0}$.
\end{proof}

This allows us to say that the distribution of the output of the basis-exchange walk in Algorithm \ref{alg1_simple} is close to $\pi$ under total variation distance. Fix a $T\in \mathbb{N}$, recall that $S_T$ is the edge set chosen by Algorithm \ref{alg:exchange_walk}. The following lemma shows rapid mixing for strongly log-concave distributions.

\begin{lemma}\label{l.small_tv_distance}
  Given $k\geq 0$. Let $\pi_T$ be the distribution of $S_T$, then there is a $c$ such that if $T = c\cdot k(\epsilon + \log n + \log (1/\delta))$, we have that $\|\pi_T - \pi\|_{TV}\leq \frac{\delta}{e^{2\epsilon} + 1}$.
\end{lemma}
\begin{proof}
  Let $P(G)$ be the basis exchange described in Algorithm \ref{alg1_simple}. By Cryan et al.~\cite{cryan2019modified} and Anari et al. \citet{anari2021log} (see Lemma \ref{l.mixing}), we just need to set 
  \begin{equation}\label{e.small_tv_distance}
    T = k\left(\log \log \left({1 \over \pi(S_0)} \right) + 2\log\left({1 \over \alpha} \right) + \log 4\right),
  \end{equation}
  then the total variation distance between $\pi_T$ and $\pi$ is at most $\alpha := \delta/(e^{\epsilon'}+1)$. Note that there are at most $N^k$ elements in $\texttt{supp}(\pi)$ where $N = O(n^2)$, and since $|E|\subset S_0$, then 
  $$\pi(S_0) = \max_{S\in {[N]\choose k}} \pi(S).$$
  Thus, $\log \log(\pi(S_0)^{-1})\leq O(\log k+\log N) = O(\log n)$. Setting $\alpha=\frac{\delta}{e^{2\epsilon}+1}$ in equation (\ref{e.small_tv_distance}) completes the proof.
\end{proof}

Finally, we analysis the run-time of Algorithm \ref{alg:exchange_walk}. In particular, we show that after proper linear time initialization, we need only $O(|E|)$ space and $O(\log n)$ time to perform each step of the basis-exchange walk, which finalize the proof of Lemma \ref{lem:main_formal}.

\textbf{Run time analysis of Algorithm \ref{alg:exchange_walk}.} We make three standard assumptions on our computation model:
\begin{enumerate}
  \item Storing the weight on each edge requires $O(1)$ space. (For example, in a floating-point style.) Note that this assumption is typically used in the non-private setting as well.

  \item 
  Tossing a biased coin (i.e., sampling from the Bernoulli distribution)  with bias $p$ needs time $O(1)$ for any $p\in (0,1)$. 
  
  \item Given $b>0$, sampling a value from the random variable with density function $\texttt{PDF}(x) = \frac{1}{2b}\exp(-|x|/b)$ needs time $O(1)$.
\end{enumerate}

Let $f(e) = \exp(\epsilon\cdot w_e)$ for any $e\in [N]$. We first prove the running time and space requirement of our algorithm, i.e, with high probability, for any $\delta\in (0,1)$, Algorithm \ref{alg:exchange_walk} runs in time $O(|E|\log (n/\delta)\log n)$, and the space complexity of Algorithm \ref{alg:exchange_walk} is $O(|E|)$. First, in line $1$ of Algorithm \ref{alg:exchange_walk}, we choose a subset $S_0$ arbitrarily, which can be implemented efficiently. 

Apart from that, the only non-trivial part lies in line 3 and line 4, where the algorithm drops an element uniformly at random and samples an element in $[N]\backslash S_t'$. Note that for any $y\in [N]\backslash S_t'$, 
$$\pi(S'_t \cup \{y\}) \propto \prod_{e\in S'_t \cup \{y\}} f(e) \propto f(y),$$
since $\prod_{e\in S'_t} f(e)$ is same for all $y\in [N]\backslash S_t'$. Recall that for any $e\notin E$, the probabilistic weight is always $f(e) = \exp(0\cdot \epsilon) = 1$, then we only need to read the weights in the edge set at the start of the algorithm, which needs time complexity of $O(|E|)$ (here, we assume that we need $O(1)$ space to store the weight of each edge). In the $t$-th step of the basis-exchange walk, we actually do the following:

\begin{enumerate}
  \item Let $E_t = ([N]\backslash S_t') \cap E$, and $\bar{E}_t = ([N]\backslash S_t')\backslash E_t$.
  \item Let $y^*$ be any element such that $y^*\in \bar{E}_t$ (we note that $f(a) = 1$ for all $a\in \bar{E}_t$). Define $\mu_t$ be the distribution on $Z_t = E_t \cup \{y^*\}$, where 
  \begin{equation*}
    \mu_t(z) \propto \left\{
    \begin{aligned}
      f(z), &\quad z\in E_t,\\
      |\bar{E}_t|, &\quad z = y^*.
    \end{aligned}
    \right.
  \end{equation*}
  \item Sample an element $z$ in $Z_t$ according to $\mu_t$.
  \item If $z\in E_t$, let $y\leftarrow z$ (in line $4$). If $z = y^*$, then sample an $e\in \bar{E}_t$ uniformly at random, and let $y\leftarrow e$.
\end{enumerate}
By repeatedly tossing a biased coin, we can sample an element from a collection of $r$ elements in time $O(\log r)$, if the probability density is already known. Thus, after the initialization, we need only $O(|E|)$ space and $O(\log n)$ time to perform each step of the basis-exchange walk.

\section{Missing Proofs in \Cref{sec:exchange_walk}}\label{app:proof_of_main}

In this section, we give the proof Theorem \ref{t.main}. Throughout this section, we actually analyze a harder case where the number of edges is confidential. Then, theorems in this section can also be easily adopted to the setting where $|E|$ is not confidential (Algorithm \ref{alg1_simple}). For the case where $|E|$ is confidential, we give the following algorithm that is very similar with Algorithm \ref{alg1_simple} except the initialization phase. The algorithm is summarized in Algorithm \ref{alg1}.

\begin{algorithm}[h]
	\caption{{Private cut approximation by $T$ steps of basis-exchange walk(with confidential $|E|$)}}\label{alg1}
	\KwIn{A graph $G\in \mathbb{R}^N_{\geq 0}$, privacy budgets $\epsilon$, $\delta$ and parameter $\beta\in (0,1)$.}
	\KwOut{A synthetic graph $\widehat{G}$.}
    Let $c = \log(1/\beta)$ and $k\leftarrow\min \left\{N, |E|+\lap(1/\epsilon) + c/\epsilon\right\}$\;
    Choose an arbitrary $S_0\subset [N]$ and $|S_0| = k$ which maximizes $|E \cap S_0|$\;
    Set $T\leftarrow O(k(\epsilon + \log n + \log (1/\delta)))$\;
    \For{$t = \{1,2,\cdots, T\}$}{
      Choose an $e\in S_{t-1}$ uniformly at random, let $S_t' \leftarrow S_{t-1}\backslash \{e\}$\;
      Choose an element $x$ in $[N]\backslash S_t'$ with probability $\propto \pi(S_t'\cup \{x\})$, let $S_{t} \leftarrow S_t'\cup \{x\}$\;
    }
    Let $\widehat{E}\subset[N]$ be the edge set corresponding to $S_T$ \;
    \For{$e\in \widehat{E}$}{
        Draw an independent Laplace noise $Z \sim \lap(1/\epsilon)$ \;
        $w_e  \leftarrow \max \{0, w_e + Z$\} \;
    }
    \For{$e'\in [N] \land e' \notin \widehat{E}$}{
        $w_{e'} \leftarrow 0$ \;
    }
    \Return{$\widehat{G} = (V,\widehat{E})$}.
\end{algorithm}

\subsection{Run time analysis of Algorithm~\ref{alg1}}\label{s.time}
Here, we show that Algorithm~\ref{alg1} runs in almost linear time and linear space. First, in line $2$ of Algorithm \ref{alg1}, if $k\geq |E|$, we then choose an arbitrary $S_0$ that covers $E$. Otherwise, we just choose an arbitrary subset of $E$ with size $k$. Both cases can be implemented efficiently. Then, when implementing the basis-exchange walk, the run time basically follows the analysis on the run time of basis exchange walk in Appendix \ref{sec:proof_main}. After obtaining the topology, since with high probability, $|\widehat{E}| = O(|E|)$. Then the loop in line $9$ completes in time $O(|E|)$. Also, the loop in line $13$ can be efficiently implemented with a proper data structure. Therefore, Algorithm \ref{alg1} does run in $O(|E|\log(n/\delta)\log n)$ time and needs $O(|E|)$ space with high probability.

\subsection{Analysis on the privacy guarantee of Algorithm \ref{alg1}}\label{sec:privacy}

In this section, we prove the privacy guarantee of Algorithm \ref{alg1}. 

\begin{theorem}\label{t.privacy}
Given $\epsilon,\delta>0$, Algorithm \ref{alg1} is $(4\epsilon,\delta)$-differentially private.
\end{theorem}

For this, we first show that for a given $k$, if the distribution of outputting a topology $S\in {[n]\choose k}$, denoted by $\pi$, satisfies
\begin{equation}\label{e.distribution}
  \pi(S) \propto \prod_{e\in S} \exp(\epsilon\cdot w_e),
\end{equation}
then $\pi$ is $(2\epsilon,0)$-differentially private (Lemma \ref{l.pi_private}). Next, we show that the distribution $\pi$ defined in \Cref{e.distribution} is {\em strongly log-concave} (see Lemma \ref{l.f_is_concave}). This is our main technical insight and it ensures that the basis-exchange walk mixes rapidly and its output distribution after $t$ iterations, denoted by ${\pi}_t$, is close to the target distribution $\pi$ in total variation distance.  Then, by a standard argument, this means that $\pi_t$ is $(\epsilon,\delta)$-differentially private. Note that without a provable guarantee on mixing time, we cannot guarantee differential privacy. 
After obtaining a private topology distributed as $\pi_t$, we use the basic composition lemma to complete the proof. The following lemma can be shown using standard techniques. 

\begin{lemma}\label{l.pi_private}
  Given $k\geq 0$, If there is an algorithm $\mathcal{A}'$ outputs a topology $S\in {[n] \choose k}$ according to the distribution in \Cref{e.distribution}, then it is $(2\epsilon,0)$-differential privacy.
\end{lemma}

\begin{proof}
[Proof of Lemma \cref{l.pi_private}]
  Let $G$ and $G'$ be a pair of neighboring graphs with $n$ vertices where there is an $i \leq {n\choose 2}$ such that $|G[i] - G'[i]|\leq 1$. Since both graphs have $n$ vertices, then the sampling space is the same. Furthermore, for any $S\in {[n]\choose k}$, 
  \begin{equation*}
    \begin{aligned}
      \frac{Pr[\mathcal{A}'(G) = S]}{Pr[\mathcal{A}'(G') = S]} &= \frac{\prod_{e\in S} \exp(\epsilon\cdot G[e])}{\prod_{e\in S} \exp(\epsilon\cdot G'[e])} \cdot \frac{\sum_{S'\in {[N]\choose k}} \prod_{e\in S'} \exp(\epsilon\cdot G'[e]) }{\sum_{S'\in {[N]\choose k}} \prod_{e\in S'} \exp(\epsilon\cdot G[e])} \\
      &\leq  \exp(2\epsilon |G[i] - G[i']|) \leq e^{2\epsilon},
    \end{aligned}
  \end{equation*}
which completes the proof.
\end{proof}

We now show that Algorithm \ref{alg1_simple} outputs a topology with distribution close to $\pi$ (in terms of the total variation distance). To do this, we first show that the target distribution $\pi$ is strongly log-concave on $\mathbb{R}_{\geq 0}^N$ (\Cref{defn:stronglogconcavity}) so that the corresponding basis-exchange walk enjoys rapid mixing. Applying Lemma \ref{l.f_is_concave} directly implies the following:
\begin{lemma}\label{l.f_is_concave_restate}
  For any graph $G\in \mathbb{R}_{\geq 0}^N$ and $k\geq 0, \epsilon>0$, the generating polynomial of distribution $\pi$ defined in  \Cref{e.distribution} is strongly log-concave at $z\in \mathbb{R}_{\geq 0}^N$.
\end{lemma}

This allows us to say that the distribution of the output of the basis-exchange walk in Algorithm \ref{alg1_simple} is close to $\pi$ under total variation distance. Fix a $T\in \mathbb{N}$, recall that $S_T$ is the edge set chosen by Algorithm \ref{alg1_simple}. Let $\epsilon' = 2\epsilon$, and $\pi_T$ be the distribution of $S_T$. Then Lemma \ref{l.small_tv_distance} in Appendix \ref{sec:proof_main} implies that there is a $c$ such that if $T = c\cdot k(\epsilon + \log n + \log (1/\delta))$, we have  
$$\|\pi_T - \pi\|_{TV}\leq \frac{\delta}{e^{\epsilon'} + 1}.$$

Finally, we use Lemma \ref{l.tv_to_dp} to compare two distributions and get an approximate DP result.

\begin{lemma}\label{l.tv_to_dp}
  Let $\mu_1$, $\mu_2$, $\mu_1'$ and $\mu_2'$ be distributions on $\Omega$. If $\mu_1$ and $\mu_2$ are a pair of $(\epsilon',\delta^*)$-indistinguishable distributions and $d_{TV}(\mu_1, \mu_1')\leq \delta'$, $d_{TV}(\mu_2, \mu_2')\leq \delta'$, then $\mu_1'$ and $\mu_2'$ are $(\epsilon', (e^{\epsilon'}+1)\delta' + \delta^*)$-indistinguishable.
\end{lemma}

\begin{proof}
  For any $S \subset \Omega$, we see that 
  \begin{equation*}
      \begin{aligned}
          Pr_{\mu_1'}(S) &\leq Pr_{\mu_1}(S) + \delta' \leq e^{\epsilon'} \cdot Pr_{\mu_2}(S) + \delta^* + \delta'\\
          &\leq e^{\epsilon'} \cdot(Pr_{\mu_2'}(S) + \delta') + \delta^* + \delta'\\
          &\leq e^{\epsilon'} \cdot Pr_{\mu_2'}(S)  + (e^{\epsilon'}+1)\delta' + \delta^*.
      \end{aligned}
  \end{equation*}
  The other side is symmetric.
\end{proof}

\begin{lemma}
  Given $k\geq 0$, Algorithm \ref{alg1_simple} with $T = c\cdot k(\epsilon + \log n + \log (1/\delta) )$ outputs a topology $S$ while preserving $(2\epsilon,\delta)$-differential privacy. 
\end{lemma}
\begin{proof}
  Combining Lemma \ref{l.pi_private} and Lemma \ref{l.small_tv_distance}, we see that for neighboring graphs, the conditions in Lemma \ref{l.tv_to_dp} are satisfied with $\delta' = \delta/(e^{\epsilon'}+1)$ and $\delta^* = 0$. Thus, by Lemma \ref{l.tv_to_dp},  Algorithm \ref{alg1_simple} outputs a topology $S$ while preserving $(2\epsilon,\delta)$-differential privacy. 
\end{proof}

\subsection{Analysis on the utility guarantee of Algorithm \ref{alg1}}

In this section, we mainly prove the following theorem on the utility of Algorithm \ref{alg1}:

\begin{theorem}\label{t.utility}
  Given any weighted graph $G = (V,E)$, $\epsilon>0$ and $\delta,\beta \in (0,1)$, then with probability at least $1-3\beta - \delta/(1+e^\epsilon)$, Algorithm \ref{alg1} outputs a synthetic graph $\widehat{G}$ with at most $|E|+\frac{2\log (1/\beta)}{\epsilon}$ edges such that 
  $$Eval(G,\widehat{G})\leq|E|\cdot \frac{3k\log(n/\beta)}{\epsilon} + \frac{6\log(n/\beta)\log (1/\beta)}{\epsilon^2}.$$
 
  Further, for spectral error we have $\left\|L_G-L_{\widehat{G}}\right\|_2 = O\left(\frac{d_{\mathsf{max}}\cdot \log(n)}{\epsilon} + \frac{\log(1/\beta)\log^2n}{n\epsilon^2}\right).$
  
\end{theorem}

We note that, in Theorem \ref{t.utility}, we do not assume the number of edges is publicly known as in Algorithm \ref{alg1}. Also, the utility guarantee of \ref{t.utility} on spectral approximation basically follows the analysis in Liu et al. \cite{liu2023optimal}, given that we have shown in Appendix \ref{sec:privacy} that in terms of the total variation distance, the distribution of the topology sampled by Algorithm \ref{alg1} is close to the distribution defined in \Cref{e.distribution}. Therefore, we only focus on the utility analysis on cut approximation.

{\noindent\it Proof of \Cref{t.main}}.
Combining analysis in Section \ref{sec:privacy}, Theorem \ref{t.utility} and the discussions in Section \ref{s.time} directly yields Theorem \ref{t.main}, by simply rescaling the privacy budget $\epsilon$ by a constant. 

In Theorem \ref{t.main}, the error bound becomes $Eval(G,\widehat{G})\leq {O}\left(\frac{{|E|\log(n/\beta)}}{\epsilon}\right)$, this is because we do not have to pay an extra additive logarithmic term on perturbing the number of edges. To complete the proof of Theorem \ref{t.utility}, we use the following fact on tail inequalities of independent Laplace noise:

\begin{fact}\label{f.tail_laplace}
  Given $k,b\geq 0$ and $k\in \mathbb{N}$. Let $\{Z_i\}_{i\in [k]}$ be $k$ i.i.d. Laplace random variable from $\texttt{Lap}(b)$. Then with probability at least $1-e^{-t}$ for any $t>0$, it holds that 
  $$\max_{i\in [k]}|Z_i| \leq (t+\log k)b.$$
  \end{fact}

\noindent Let $$\mathcal{K}_{\text{good}} = \left\{k\in \mathbb{N} :~|E|\leq k \leq |E|+{2\log(1/\beta)\over \epsilon}\right\}.$$ 
Recall that $k = \min \left\{N, |E|+\lap(1/\epsilon) + {\log(1/\beta) \over \epsilon}\right\}$, then from the tail inequality of Laplace distribution, we have that 
$$Pr[k\in \mathcal{K}_{\text{good}}] \geq 1-\beta.$$

Next, we consider the case where $k\in \mathcal{K}_{\text{good}}$. Given $G = ([n],E)$ and any topology $S\in {[N]\choose k}$ (recall that $N = {n\choose 2}$), we write $G|S\in \mathbb{R}_{\geq 0}^N$ to denote the restriction of $G$ on $S$. That is, $G|S$ is a graph with $n$ vertices whose edge set is $S\cap E$, and $w_e(G|S) = w_e(G)$ for any $e\in S\cap E$. Next, we show that the topology $S_T$ given by Algorithm \ref{alg1} is good:
\begin{lemma}\label{l.topology_not_bad}
In Algorithm \ref{alg1} with $\delta \in (0,1)$, given $k\in \mathcal{K}_{\text{good}}$, then with probability at least $1-\beta - \delta/(1+e^\epsilon)$, it holds that 
$$\|(G|S_T) - G \|_1 \leq k\log(n/\beta).$$
\end{lemma}
\begin{proof}
Let $f(e) = \exp(\epsilon\cdot w_e)$ for any $e\in [N]$.   For any $t>0$ and $k\in \mathcal{K}_{\text{good}}$, let 
$$\mathcal{S}_t = \left\{S \subset [N]:|S| = k \land   \|(G|S) - G\|_1\geq t\right\}.$$ Recall that the distribution $\pi$ defined on ${[N]\choose k}$ satisfies 
  \begin{equation}
    \pi(S) \propto \prod_{e\in S} f(e) =\prod_{e\in S} \exp(\epsilon\cdot w_e).
  \end{equation}

Therefore, let $S_{OPT}\in {[N]\choose k}$ be any topology that $E\in S_{OPT}$ (note that $k\geq|E|$ if $k\in \mathcal{K}_{\text{good}}$). Then we have that 
  \begin{equation*}
    \begin{aligned}
      \pi(\mathcal{S}_t) \leq \frac{\pi(\mathcal{S}_t)}{\pi(S_{OPT})} &\leq  {N\choose k}\cdot \frac{\max_{S\in \mathcal{S}_t}\prod_{e\in S}f(e)}{\prod_{e\in E}f(e)}\\
      & = {N\choose k}\cdot \frac{\max_{S\in \mathcal{S}_t}[\exp(0\cdot \epsilon)]^{|S\backslash E|}\prod_{e\in S\cap E}f(e)}{\prod_{e\in E}f(e)}\\
      & = {N\choose k}\cdot \left[ \max_{S\in \mathcal{S}_t} \prod_{e\in E\backslash S}f^{-1}(e)\right] \\
      &= {N\choose k} \exp(-\epsilon \|(G|S) - G\|_1)\\
      &\leq \exp(k\log n - t) \leq \beta
    \end{aligned}
  \end{equation*}
  if we set $t = k\log(n/\beta)$. Let the distribution of $S_T$ outputted by Algorithm \ref{alg1} be $\pi_T$. By Lemma \ref{l.small_tv_distance}, we have that $\|\pi - \pi_T\|_{TV}\leq \frac{\delta}{e^\epsilon+1}$. That is to say: 
  $$\pi_T(\mathcal{S}_t) \leq \beta + \frac{\delta}{e^\epsilon+1},$$
  where $t = k\log(n/\beta)$. This completes the proof of Lemma~\ref{l.topology_not_bad}.
\end{proof}

{\noindent \it Proof of \Cref{t.utility}.} 
By Fact \ref{f.tail_laplace}, we further have that with probability at least $1-\beta$, it holds that 
  $$\left\|(G|S_T) - \widehat{G}\right\|_1 \leq k\frac{\log (n^2/\beta)}{\epsilon} \leq 2k\frac{\log (n/\beta)}{\epsilon}.$$
  Therefore, using the union bound, we have that with probability at least $1-2\beta$ for some small $\beta$, it holds that $\left\|(G|S_T) - \widehat{G}\right\|_1 \leq 2k \log(n/\beta)/\epsilon$ for some $k\leq |E|+2\log (1/\beta)/\epsilon$. Therefore, we have that 
  \begin{equation}
    \begin{aligned}
      Eval(G,\widehat{G}) &= \max_{q\in [0,1]^N} |q^\top G - q^\top\widehat{G}| = q^\top |G  -\widehat{G}| \\
      & \leq \sum_{i\in [N]}|G[i] - \widehat{G}[i]| = \|G - \widehat{G}\|_1 \\
      & \leq \|G - (G|S_T)\|_1 + \|(G|S_T) - \widehat{G}\|_1 \\
      & \leq k\cdot \left(\frac{\log(n/\beta)}{\epsilon} + \frac{2\log(n/\beta)}{\epsilon}\right) \\
     &= |E|\cdot \frac{3k\log(n/\beta)}{\epsilon} + \frac{6\log(n/\beta)\log (1/\beta)}{\epsilon^2}
    \end{aligned}
  \end{equation}
  holds with probability at least $1- 3\beta - \delta/(1+e^\epsilon)$ due to the union bound.

\section{Missing proofs in Section \ref{sec:high_pass}}
\label{sec:high_pass_proofs}
\subsection{Proof of Theorem \ref{t.ana_on_alg2}}
Here, we first restate Theorem \ref{t.ana_on_alg2}.
\begin{theorem}\label{t.ana_on_alg2_restate}
  For any $\epsilon>0$ and $\delta\in (0,1)$, Algorithm \ref{alg2} preserves $(\epsilon,\delta)$ differential privacy. Also, with probability at least $1-\delta$, Algorithm \ref{alg2} outputs a $\widehat{G}$ such that 
  $$Eval(G,\widehat{G}) =  O\left(\frac{|E|\log (n/\delta)}{\epsilon}\right).$$
\end{theorem}

\begin{proof}
  We first show the privacy guarantee. Given any graph $G = ([n],E)$, we consider two types of neighboring graphs: 
  \begin{enumerate}
    \item $G_1 = ([n],E_1)$ has an extra edge $e'$ with weight $1$ outside the edge set $|E|$, and 
    \item $G_2 = ([n],E_2)$ has edge set $E_2 \subset E$.
  \end{enumerate}
  Firstly, by the tail inequality of Laplace distribution (Lemma \ref{f.tail_laplace}), we see that by setting $t = \frac{2\log(2n/\delta)}{\epsilon}$, for every edge with weight less or equal than $1$, it holds with probability at least $1-\frac{\delta}{n^2}$ that this edge will be filtered. Let the extra edge in $G_1$ be $e'$, then by the union bound, with probability at least $1-\delta $, this edge will be set to $0$. Suppose it happens, then Algorithm \ref{alg2} preserves $(\epsilon,0)$ differential privacy due to the Laplace mechanism and the post-processing property of differential privacy. For case 2 where $G_2$ has the same or smaller edge set, clearly Algorithm \ref{alg2} also preserves $(\epsilon,\delta)$ differential privacy. Thus, Algorithm \ref{alg2} preserves $(\epsilon,\delta)$ differential privacy for $\epsilon>0$ and $0<\delta<1$.

  For the utility part, again by the tail inequality of Laplace distribution and the union bound, we see that the difference between $G$ and $\widehat{G}$ in terms of $\ell_1$ norm is 
  $$Pr\left[\left\|G - \widehat{G}\right\|_1 \leq 2|E|\cdot  \frac{2\log(2n/\delta)}{\epsilon}\right] \geq 1-\delta,$$
  and this means that 
  \begin{equation}
    \begin{aligned}
      Eval(G,\widehat{G}) &= \max_{q\in [0,1]^N} |q^\top G - q^\top\widehat{G}| = q^\top |G  -\widehat{G}| \\
      &\leq \sum_{i\in [N]}|G[i] - \widehat{G}[i]| = \|G - \widehat{G}\|_1 \\
      & \leq 4|E|\cdot  \frac{\log(2n/\delta)}{\epsilon}
    \end{aligned}
  \end{equation}
  holds with probability at least $1-\delta$. This completes the proof of Theorem~\ref{t.ana_on_alg2}.
\end{proof}

\noindent In particular, if the unweighted degree is bounded by $d_{\mathsf{max}} \leq n-1$, Algorithm \ref{alg2} will have a more accurate approximation on the size of small cuts. For this property, we give the following theorem:

\begin{theorem}\label{t.small_cut_alg2}
  Given $\epsilon>0$ and $0<\delta<1$. For any $G=([n],E)$ with maximum unweighted degree $d_{\mathsf{max}}\leq n-1$, with probability at least $1-\delta$, Algorithm \ref{alg2} outputs a $\widehat{G}$ such that for any $S\in 2^{[n]}$, 
    $$|\Phi_G(S) - \Phi_{\widehat{G}}(S)| \leq  \frac{4 d_{\mathsf{max}} |S|\log (2 n/\delta)}{\epsilon}.$$
\end{theorem}

\begin{proof}
    In Algorithm \ref{alg2}, we first add a Laplace noise on each edge in the edge set, then filter the edges with weights below or equal to the threshold $t$. Thus, for any $e\in E$, we have that with probability at least $1-\delta$,
    $$|w_e - \widehat{w}_e|\leq 2t,$$
    while for every $e\notin E$, $w_e = \widehat{w}_e = 0$. Consider any $S\in 2^{[n]}$, since $E(S,[n]\backslash S)$ has at most $d_{\mathsf{max}} |S|$ edges, then we have that 
      $$|\Phi_G(S) - \Phi_{\widehat{G}}(S)|\leq 2t\cdot d_{\mathsf{max}} |S| = \frac{4d_{\mathsf{max}}|S|\log(2n/\delta)}{\epsilon}.$$
      This concludes the proof of Theorem~\ref{t.small_cut_alg2}.
\end{proof}

Note that, For any disjoint $S,T\in 2^{[n]}$, there are at most $d_{\mathsf{max}}\cdot \min(|S|,|T|)$ edges crossing through the $(S,T)$-cut. Thus, we immediately have the following corollary:

\begin{corollary}
      Given $\epsilon>0$ and $0<\delta<1$. For any $G=([n],E)$ with maximum unweighted degree $d_{\mathsf{max}}\leq n-1$, with probability at least $1-\delta$, Algorithm \ref{alg2} outputs a $\widehat{G}$ such that for any disjoint $S,T\in 2^{[n]}$, 
  $$|\Phi_G(S,T) - \Phi_{\widehat{G}}(S,T)| = O\left(\frac{d_{\mathsf{max}}\cdot \min(|S|,|T|)\log (n/\delta)}{\epsilon}\right).$$
\end{corollary}

\subsection{Proof of Theorem \ref{t.spectral}}
Here, we first restate Theorem \ref{t.spectral}.

\begin{theorem}\label{t.spectral_restate}
  For any $G=([n],E)$ with bounded maximum unweighted degree $d_{\mathsf{max}}\leq n-1$, with high probability, Algorithm \ref{alg2} outputs a $\widehat{G} = (V,\widehat{E})$ such that for all $x\in \mathbb{R}^n$ with $\|x\|_2 \leq 1$, 
  $$|x^\top L_Gx - x^\top L_{\widehat{G}}x| =  O\left(\frac{d_{\mathsf{max}}\log(n/\delta)}{\epsilon}\right).$$
\end{theorem}
\begin{proof}
Recall that we set the threshold $t = {c\log(n/\delta)\over \epsilon}$ in Algorithm \ref{alg2}. By the analysis in Theorem \ref{t.ana_on_alg2}, we have that  
$$|w_e - \widehat{w}_e|\leq 2t, \text{for all } e\in E$$
with probability at least $1-\delta$.

We use Lemma~\ref{l.spectral} to convert the approximation on $(S,T)$-cuts to the approximation on the spectrum. 
 Let $L_{\widetilde{G}} = L_G - L_{\widehat{G}}$. Clearly $L_{\widetilde{G}}\in \mathbb{R}^{n\times n}$ is a symmetric matrix such that $L_{\widetilde{G}} \mathbf{1} = \mathbf{0}$. We then check the conditions in Lemma \ref{l.spectral} for $L_{\widetilde{G}}$. First, with probability at least $1-\delta$, we have that 
\begin{equation}\label{e.4.31}
  \max_{i\in [n]} |L_{\widetilde{G}}[i,i]| = \max_{i\in [n]}|\Phi_{G}(\{i\}) - \Phi_{\widehat{G}}(\{i\})| \leq 2t\cdot \min\{d_{\mathsf{max}}, |E|\} \leq 2t\cdot d_{\mathsf{max}}.
\end{equation}
Since $\mathbf{1}_S[i]\mathbf{1}_T[j]$ and $\mathbf{1}_S[j]\mathbf{1}_T[i]$ cannot  be $1$ simultaneously, for all $\mathbf{1}_S,\mathbf{1}_T\in \{0,1\}^n$ such that $S\cap T = \varnothing$, we have 
\begin{equation}\label{e.4.32}
  \begin{aligned}
    \frac{|\mathbf{1}_S^\top L_{\widetilde{G}}\mathbf{1}_T|}{\|\mathbf{1}_S\|_2 \|\mathbf{1}_T\|_2} &= \frac{1}{{\|\mathbf{1}_S\|_2 \|\mathbf{1}_T\|_2}}\sum_{\{i,j\}\in E} w_{\{i,j\}} |\mathbf{1}_S[i]\mathbf{1}_T[i] + \mathbf{1}_S[j]\mathbf{1}_T[j] - \mathbf{1}_S[i]\mathbf{1}_T[j] - \mathbf{1}_S[j]\mathbf{1}_T[i]| \\
    & = \frac{|\Phi_{\widetilde{G}}(S,T)|}{{\|\mathbf{1}_S\|_2 \|\mathbf{1}_T\|_2}} = \frac{|\Phi_{{G}}(S,T) - \Phi_{\widehat{G}}(S,T)|}{\|\mathbf{1}_S\|_2 \|\mathbf{1}_T\|_2} \\
    &\leq \frac{1}{\|\mathbf{1}_S\|_2 \|\mathbf{1}_T\|_2}\cdot 2t\cdot \min\left\{d_{\mathsf{max}}\cdot \min \left\{|S|,|T|\right\},{|E|}\right\} \\
    &= 2t\cdot \min\left\{d_{\mathsf{max}}\cdot \min\left\{\sqrt{\frac{|S|}{|T|}}, \sqrt{\frac{|T|}{|S|}}\right\},\frac{|E|}{\sqrt{|S|\cdot |T|}}\right\}  \leq 2t\cdot d_{\mathsf{max}}.
  \end{aligned}
\end{equation}
holds simultaneously.  
We now let the $\alpha$ in Lemma \ref{l.spectral} be $\alpha = 2td_{\mathsf{max}}$.
Combing the fact that $L_{\widetilde{G}} \mathbf{1} = \mathbf{0}$ and Equation (\ref{e.4.31}) implies that $\|L_{\widetilde{G}}\|_{\infty} \leq 2\alpha$, and all diagonal entries of $A$ has absolute value $\alpha$ which is less than $\alpha(\log (\ell/\alpha) +1)$. Here, $\ell$ is the infinity norm of $L_{\widetilde{G}}$, and it is bounded by $2\alpha$. 
Also, Equation \ref{e.4.32} implies that for all non-zero vectors $u,v\in \{0,1\}^n$ and $u^\top v  =0$, it holds that
$$\frac{|u^\top L_{\widetilde{G}} v|}{\|u\|_2\|v\|_2} \leq \alpha.$$
Thus, Theorem \ref{t.spectral} follows by applying Lemma \ref{l.spectral}. with $\alpha = 2td_{\mathsf{max}} = O(d_{\mathsf{max}}\log(n/\delta)/\epsilon)$.  
\end{proof}

\section{Optimality of our algorithms}\label{app:optimality}
We have proposed our algorithms for approximating cut queries by actually preserving the answers to all linear queries on the given graph. In this section, we show that for answering any linear queries in $[0,1]^N$ by a synthetic graph $\widehat{G}$ under $(\epsilon,\delta)$-differential privacy, the error w.r.t. $G = (V,E)$ is at least $\Omega((1-\delta)|E|/(e^\epsilon + 1))$ for $\epsilon>0$ and $0<\delta<1$. Thus, both our $(\epsilon, \delta)$-differentially private algorithms for answering linear queries are optimal.
\begin{theorem}\label{t.lb_on_sparse}
  Recall that $N = {n\choose 2}$. Let $\mathcal{M}:\mathbb{R}_{\geq 0}^N\rightarrow \mathbb{R}^N$ be a $(\epsilon, \delta)$-algorithm such that for any given graph $G$, it outputs a $\widehat{G}$ which satisfies $$\mathbb{E}[Eval(G,\widehat{G})] = \mathbb{E} \left[\max_{q\in [0,1]^N}|q(G)- (\widehat{G})|\right] \leq \alpha,$$ then $\alpha = \Omega\left(\frac{(1-\delta)|E|}{e^\epsilon + 1}\right)$.
\end{theorem}
\begin{proof}
  We first note that preserving the $\ell_1$ norm of $G$ can be reduced to answering any collection $\mathcal{Q}$ of linear queries in $[0,1]^N$. (It is easy to verify the opposite direction is also correct). Let $\widehat{G}$ be the output. Then we construct two linear queries 
  \begin{equation*}
    q_+ := \left\{\begin{aligned}
      1, &\quad i\in [N] \text{ and }G[i]\geq \widehat{G}[i], \\
      0, &\quad i\in [N] \text{ and }G[i]< \widehat{G}[i];
    \end{aligned} \right.
  \end{equation*}
  and $q_- = \{1\}^N - q_+$, which flips the answer of $q_+$. Then we see that 
  $$\langle q_+, (G - \widehat{G}) \rangle + \langle{q_-}, (G-\widehat{G})\rangle = \|G - \widehat{G}\|_1,$$
  where $\langle \cdot, \cdot \rangle$ is the inner product of two vectors. Since $q_+,q_-\in [0,1]^N$, then 
  \begin{equation}\label{eq.lower_bound}
   \begin{split}
    Eval(G,\widehat{G}) &= \max_{q\in [0,1]^N}|q(G - \widehat{G})| \\
    &\geq \max \{\langle q_+, (G - \widehat{G}) \rangle, \langle{q_-}, (G-\widehat{G})\rangle\} \\
    &\geq \|G-\widehat{G}\|_1/2.   
   \end{split} 
  \end{equation}
  Now we aim to find the lower bounds on minimizing $\|G - \widehat{G}\|_1$. Consider the dense case where $|E| = N$, there is a widely accepted sense in which the Laplace mechanism, i.e., adding a random perturbation on each pair of vertices achieves optimal error on preserving $G$ in terms of $\ell_1$ norm for mechanisms adding oblivious noise (see \cite{aumuller2022representing} and Theorem 6 in \cite{koufogiannis2015optimality}). Since each magnitude of the i.i.d. Laplace noise added by the Laplace mechanism is $O(1/\epsilon)$, then the error in terms of $\ell_1$ norm is $O(N\log n/\epsilon) = O(|E|\log n/\epsilon)$. Here, the $\log n$ factor is due to the union bound. For general algorithms in the sparse case, we prove the following lemma for the sake of completeness:

  \begin{lemma}\label{l.lb_on_l1norm}
    Given $\epsilon>0$ and $0<\delta<1$. Let $\mathcal{M}:\mathbb{R}_{\geq 0}^N\rightarrow \mathbb{R}^N$ be a $(\epsilon, \delta)$-algorithm such that for any given graph $G$, it outputs a $\widehat{G}$ which satisfies $\mathbb{E}\left[\|G - \widehat{G}\|_1\right] \leq \alpha$, then $\alpha = \Omega\left(\frac{(1-\delta)|E|}{e^\epsilon + 1}\right)$. 
  \end{lemma}

\begin{proof}
  We first consider the dense case where $|E| = N$, and the goal is to show that the error $\alpha$ is at least $\Omega(N)$. Let $\mathcal{M}_i:\mathbb{R}_{\geq 0}^N \rightarrow \mathbb{R}$ be the projection of $\mathcal{M}$ on the $i$-th position for $i\in [N]$, namely $\widehat{G}[i] = \mathcal{M}_i(G)$ if $\mathcal{M}(G) = \widehat{G}$. By the post-processing lemma, each $\mathcal{M}_i$ is simultaneously $(\epsilon,\delta)$-differential privacy. For any $G\in \mathbb{R}_{\geq 0}^N$, let $G^{(i)}$ for $i\in [N]$ be $N$ different neighboring graphs of $G$ such that $G^{(i)}[i] = G[i]+1$. Let $0<\beta <1-\delta$ be a parameter. Suppose that for any fixed $i\in [N]$, with probability at least $1-\beta$, the error $|G[i] - \widehat{G}[i]|$ is at most $\alpha_i$, where $\alpha_i$ is any constant satisfies $0< \alpha_i<1/2$. Since $\mathcal{M}_i$ is $(\epsilon,\delta)$-differentially private, then 
  \begin{equation}
    \begin{aligned}
      \beta &\geq Pr[|\mathcal{M}_i(G^{(i)}) - G^{(i)}[i]|> \alpha_i] \\
      &\geq Pr[\mathcal{M}_i(G^{(i)}) <G^{(i)}[i]- \alpha_i] \\
      & \geq Pr[\mathcal{M}_i(G^{(i)}) <G^{(i)}[i]-1/2]\\
      &\geq Pr[\mathcal{M}_i(G^{(i)}) \leq G[i]+\alpha_i] \\
      & \geq e^{-\epsilon} \cdot (Pr[\mathcal{M}_i(G) \leq G[i]+\alpha_i]-\delta)\\
      &\geq (1-\beta-\delta)\cdot e^{-\epsilon}.
    \end{aligned}
  \end{equation}
  This inequality is violated for any $\beta <\frac{e^{-\epsilon}(1-\delta)}{1+e^{-\epsilon}}$. Therefore, for all $0<\epsilon,\delta<1$, if we choose a $\beta$ satisfying $0.99\cdot \frac{1-\delta}{e^\epsilon+1}<\beta<\frac{e^{-\epsilon}(1-\delta)}{1+e^{-\epsilon}}<1-\delta$, then we have that for any $i\in [N]$,
  $$Pr\left[|G[i] - \widehat{G}[i]| > \alpha_i \right] >\beta>0.99\cdot \frac{1-\delta}{e^\epsilon +1}.$$
  Then, the expected error on $\ell_1$ norm will be 
  $$\mathbb{E} \left[ \left\|G-\widehat{G} \right\|_1 \right] \geq \sum_{i\in [N]} \alpha_i \cdot Pr \left[|G[i] - \widehat{G}[i]| > \alpha_i \right] \geq 0.99\cdot \sum_{i\in [N]}\frac{\alpha_i(1-\delta)}{e^\epsilon+1},$$ 
  which will be at least $\Omega\left({(1-\delta)N \over (e^\epsilon + 1)}\right)$.

  We now move to the sparse case, i.e.,  $|E| = o(N)$. For the sake of contradiction, suppose there is an algorithm that outputs a graph, $\widehat{G}$, such that $\mathbb{E}\left[\|\widehat{G} - G\|_1 \right] = o(|E|)$, then consider the density case with $N' = \Theta(|E|)$, we can always borrow some dimensions with zero weight to make it a sparse case with respect to $N$, where $N' = o(N)$. The above argument shows that the error for the dense case is at least $\Omega(N') = \Omega(|E|)$. Note that we construct the sparse case by borrowing new dimensions, and it will only increase the error in terms of $\ell_1$ norm. That is, the error in terms of $\ell_1$ norm in the sparse case w.r.t. $N$ is also at least $\Omega(|E|)$, which leads to a contraction and concludes the proof of Lemma~\ref{l.lb_on_l1norm}. 
\end{proof}
\noindent Combining Equation (\ref{eq.lower_bound}) and Lemma \ref{l.lb_on_l1norm} completes the proof of Theorem \ref{t.lb_on_sparse}.
\end{proof}

\section{Extension to Continual Observation}\label{app:contious_observation}\label{sec:contious_observation}

In previous sections, we mainly discuss static algorithms on confidential graphs. However, there is a long line of work on the dynamic setting known as continual observation model \citep{chan2011private,dwork2010differentially}. Here, we introduce a general setting for private graph analysis under continual observation. Consider that we start from an empty graph $\widetilde{G}_0 = G_{\text{empty}} = ([n],\varnothing)$. Given a time upper bound $T$, a private algorithm $\mathcal{M}$ on graphs and a utility function $s:\mathbb{R}_{\geq 0}^{n\choose 2}\times \mathbb{R}_{\geq 0}^{n\choose 2} \rightarrow \mathbb{R}_{\geq 0}$, in each round $i = 1,\cdots, T$:
\begin{enumerate}
  \item An oblivious adversary propose an update $G_i = (e,w)\in [{n \choose 2}]\times \mathbb{R}_{\geq 0}$, and let $\widetilde{G}_{i} = \widetilde{G}_{i-1} + G_i$, which means increasing the weight of the $e$-th pair of vertices by $w$. 
  \item $\mathcal{M}$ outputs a synthetic graph that approximates $\widetilde{G}_{i}$ with respect to the utility function $s$.
\end{enumerate}

To apply Definition \ref{d.dp} in continual observation, we view streams as the input datasets and two streams are neighboring if they differ in one update by at most $1$ (on same edge $e$). This is also called the {\em event-level} privacy.  We note that under such a privacy notion, two neighboring streams could still have different topology (consider two updates $(e,0)$ and $(e,1)$ where $e \leq {n \choose 2}$).

\textbf{The algorithm.} 
We use the binary mechanism  \citep{chan2011private, dwork2010differentially}. A trivial approach is to initialize ${n\choose 2}$ copies of binary mechanism on each pair of vertices, which needs $O(n^2)$ time per round. Our algorithm speeds up the procedure using the sparsity of partial sums of graphs. In Algorithm \ref{alg.continuous_observation}, we use $G_t$ ($t\in [T]$) to denote the graph whose only edge is the $t$-th coming edge (with weight $w_t$) in the stream. We use $G_{\text{empty}}$ to denote the empty graph.

\begin{algorithm}[t]
	\caption{Binary mechanism (in the graph setting)}\label{alg.continuous_observation}
	\KwIn{A sequence of edges $\{(e_i,w_i)\}_{i\in [T]}$, a mechanism $\mathcal{M}:\mathbb{R}_{\geq 0}^{n\choose 2} \rightarrow \mathbb{R}_{\geq 0}^{n\choose 2}$.}
	\KwOut{A sequence of synthetic graphs $\{\widehat{G}_i\}_{i\in [T]}$.}

    \For{$t = \{1,2,\cdots, T\}$}{
      Express $t$ in its binary form $t = \sum_{l\geq 0} k_l(t) \cdot 2^{l}$\; 
      \tcp{Here, $k_l(t)\in \{0,1\}$ is the $(l+1)$-th bit of $t$ (in binary form) starting from the right hand.}
      Let $j \leftarrow \min_{l\geq 0} \{k_l(t)\neq 0\}$\;
      Let $\alpha_{j} \leftarrow \sum_{l<j} \alpha_l + G_t$\;
      \For{$0\leq l \leq j-1$}{
        $\alpha_i = G_{\text{empty}}$
      }
      Run $\mathcal{M}$ on $\alpha_j$, and set $\widehat{\alpha}_j = \mathcal{M}(\alpha_j)$ \;
      Release $\widehat{G}_t = \sum_{l\geq 0}\widehat{\alpha}_l$ \;
    }
\end{algorithm}

\begin{theorem}
Given $T\in \mathbb{N}$, for any stream $S\in ([{n\choose 2}]\times \mathbb{R}_{\geq 0})^T$ with $m$ be the number of edges in the final state $\widetilde{G}_T$, let $\widetilde G_t$ denote the graph formed by the stream until time $t$. Then for any $\epsilon>0$ and $0<\delta<1$, Algorithm \ref{alg.continuous_observation} is an $O(\log T)$ amortized update time, $(\epsilon,(T+1)\delta)$ differentially private algorithm under continual observation that at every time epoch, outputs a graph $\widehat G_t$, such that with probability at least $1 -T\delta$,
$$\max_{t\in [T]} Eval\left(\widehat{G}_t, \widetilde{G}_t\right) = O \left(\frac{m\log ({n\over \delta})  \sqrt{\log^3(T) \log({1\over \delta})}}{\epsilon}\right).
$$

\end{theorem}


We first give the privacy guarantee of Algorithm \ref{alg.continuous_observation}, which is a direct application of the advanced composition lemma (see Lemma \ref{l.adv_composition}).
\begin{theorem}\label{t.continuous_pri}
  Given that $\mathcal{M}:\mathbb{R}_{\geq 0}^{n\choose 2} \rightarrow \mathbb{R}_{\geq 0}^{n\choose 2}$ is a $(\epsilon,\delta)$ edge level differentially private algorithm, then Algorithm \ref{alg.continuous_observation} (with $\mathcal{M}$ as input) preserves $(\epsilon',\delta'+T\delta)$ edge level differential privacy for any $\delta'\in (0,1)$, where $$\epsilon' = O\left(\epsilon\cdot \sqrt{\log(T)\log (1/\delta')}\right).$$
\end{theorem}

\begin{proof}
  (Sketch.) Let $S = \{(e_i^S,w_i^S)\}_{i\in [T]}$ and $S' = \{(e_i^{S'},w_i^{S'})\}_{i\in [T]}$ be a pair of neighboring streams of length $T$ that only differs in one item by $1$. Let $t$ be the time such that $w_t^S \neq w_t^{S'}$ and $|w_t^S - w_t^{S'}|\leq 1$. Note that in Algorithm \ref{alg.continuous_observation}, the $t$-th edge can occur in at most $O(\log T)$ different partial sums. Since the distribution of the released partial sums formed by graphs that don't contain the $t$-th edge is identical between two neighboring inputs, then every output in time $i\in [T]$ can be considered as the outcome of post-processing over the released partial sums formed by graphs that contain the $t$-th edge. Therefore, Theorem \ref{t.continuous_pri} can be proved by applying the advanced composition lemma.
\end{proof}

\noindent For the accuracy part, we show that if the utility function $s$ is ``sub-additive'', then it's possible to have a good utility bound.

\begin{definition}
  [Sub-additive measure] Let $s:\mathbb{R}_{\geq 0}^{{n\choose 2}}\times \mathbb{R}_{\geq 0}^{{n\choose 2}} \rightarrow \mathbb{R}_{\geq 0}^{{n\choose 2}}$ be a symmetric measure. We say that $s$ is sub-additive if for any $\ell\in \mathbb{N}_{\geq 0}$ and $G_1,\cdots,G_\ell,G_{\ell+1}, G_{2\ell}\in \mathbb{R}_{\geq 0}^{{n\choose 2}}$, it holds that 
  $$s\left(\sum_{i=1}^\ell G_i, \sum_{j = \ell+1}^{2l}G_j\right) \leq \sum_{i=1}^{l}s(G_i,G_{i+l}).$$
\end{definition}
\noindent For any $t\in [T]$, let $\widetilde{G}_t = \sum_{i\in [t]}G_i$. For any $G\in \mathbb{R}_{\geq 0}^N$, we also define 
$$[G] = \left\{G':~\forall i\in [N], G'[i]\leq G[i] \right\}.$$ We now discuss the utility of Algorithm \ref{alg.continuous_observation}. Note that in each round, the output is simply the sum of synthetic graphs given by the private algorithm $\mathcal{M}$. For $t\in [T]$ and $0\leq i \leq \lceil \log T\rceil$, let $\alpha_i^{(t)}$ be the graph represented by the partial sum $\alpha_i$ in Algorithm \ref{alg.continuous_observation} at round $t$ (in a same way, we define $\widehat{\alpha}_i^{(t)}$ be the perturbed partial sum). Suppose the measurement $s$ that we are interested in is sub-additive. Since there are at most $\lceil \log T\rceil$ such partial sums (with corresponding synthetic graphs), then the error on any round $t\in [T]$ is 
\begin{align*}
    s\left(\widehat{G}_t, \widetilde{G}_t\right) &= s\left( \sum_{i = 0}^{\lceil \log T\rceil } \widehat{\alpha}_i^{(t)}, \sum_{i = 0}^{\lceil \log T\rceil } \alpha_i^{(t)}\right) \\
    &\leq \log T \left(\max_{G\in [\widetilde{G}_T]} \texttt{err}(G)\right).
\end{align*}
Thus, we directly give the following theorem on utility:
\begin{theorem}\label{t.continuous_utility}
  Let $s:\mathbb{R}_{\geq 0}^{{n\choose 2}}\times \mathbb{R}_{\geq 0}^{{n\choose 2}} \rightarrow \mathbb{R}$ be any sub-additive measure. If $\mathcal{M}:\mathbb{R}_{\geq 0}^{n\choose 2} \rightarrow \mathbb{R}_{\geq 0}^{n\choose 2}$ is a $(\epsilon,\delta)$ edge level differentially private algorithm and for any $G$, with probability at least $1-\beta$, it outputs $\widehat{G}$ such that $s(G,\widehat{G}) \leq \texttt{err}(G)$, then if Algorithm \ref{alg.continuous_observation} outputs $\widehat{G}_1,\cdots, \widehat{G}_T$, it holds with probability at least $1-T\beta$ such that 
  $$\max_{t\in [T]} s\left(\widehat{G}_t, \widetilde{G}_t\right) \leq \log T \left(\max_{G\in [\widetilde{G}_T]} \texttt{err}(G)\right).$$
\end{theorem}

\begin{remark}
  In Theorem \ref{t.continuous_utility}, we see that we need to set $\beta = o(1/T)$ so that we can get a high probability bound.
\end{remark}

Now, we substitute the static algorithm $\mathcal{M}$ by Algorithm \ref{alg2} and let $s$ be $s(\widehat{G}_t, \widetilde{G}_t) = Eval(\widehat{G}_t, \widetilde{G}_t)$ for any $t$. By specifying parameters, we have the following corollary, by the fact that $Eval(\cdot,\cdot)$ is obviously sub-additive.

\begin{corollary}
Given $T\in \mathbb{N}$, $\epsilon>0$ and $0<\delta<1$, for any stream $S\in ([{n\choose 2}]\times \mathbb{R}_{\geq 0})^T$ with $m$ be the number of edges in the final state $\widetilde{G}_T$, there exists an $O(\log(T))$ amortized run-time $(\epsilon,(T+1)\delta)$ differentially private algorithm under continual observation such that with probability at least $1-\delta T$, it holds that:
 $$\max_{t\in [T]} Eval\left(\widehat{G}_t, \widetilde{G}_t\right) \leq O\left((\log T)^{1.5} \cdot \left(\frac{m\log (n/\delta)\sqrt{\log(1/\delta)}}{\epsilon}\right)\right).$$
\end{corollary}

\begin{proof}
In each resample procedure, we call the static algorithm $\mathcal{M}$ (Algorithm \ref{alg2}) with parameter
$$\epsilon_0 = \frac{\epsilon}{\sqrt{\log T \log(1/\delta)}} \text{  and  } \delta_0 = \delta.$$
By Theorem \ref{t.continuous_pri}, we see that Algorithm \ref{alg.continuous_observation} preserves $(\epsilon, (T+1)\delta)$-differential privacy under a stream of length $T$. By the utility guarantee of Algorithm \ref{alg2} (Theorem \ref{t.ana_on_alg2}), we have that 

\begin{equation*}
  \begin{aligned}
    \max_{t\in [T]} Eval\left(\widehat{G}_t, \widetilde{G}_t\right) &\leq \log T \left(\max_{G\in [\widetilde{G}_T]} \texttt{err}(G)\right) \\
    & = \log T \cdot \left(\frac{m \log (n/\delta_0)}{\epsilon_0}\right)\\
    & = O\left((\log T)^{1.5} \cdot \left(\frac{m\log (n/\delta)\sqrt{\log(1/\delta)}}{\epsilon}\right)\right)
  \end{aligned}
\end{equation*}
holds with probability at least $1-T \delta$ due to the union bound (recall that Algorithms \ref{alg2} succeeds with probability at least $1-\delta$).

For the time complexity, let's assume $T = 2^a$ for some $a\in \mathbb{N}$ for simplicity. In this case, we have to compute $T = 2^a$ synthetic graphs, and there are $a$ layers in the binary mechanism (Algorithm \ref{alg.continuous_observation}). In the $i$-th layer (from bottom to top), each partial sum contains at most $2^i$ edges, but there are $T/2^i$ nodes. By the time guarantee of Algorithm \ref{alg2}, for an $m$ edges graph, the algorithm needs $O(m)$ time even if $m$ is not publicly known. Thus, to compute all synthetic graphs corresponding to the nodes in all layers, the time needed is 
$$\sum_{i=0}^{a} O(2^i) \cdot {T \over 2^i} = O(T\log_2 T).$$
Since there are $T$ rounds in total, then the average run-time for each round is $O(\log T)$. Note that the case where $\log_2 T\notin \mathbb{N}$ can be reduced to $T' = 2^{\lceil \log_2 T\rceil} \leq 2T$, which completes the proof.
\end{proof}

\section{A detailed comparison with other approaches}\label{s.compare}

In this section, we briefly summarize previous techniques on differentially private cut approximation, and make a detailed comparison between the utility (in addition to the resource required) from those techniques and ours (see Table \ref{t1}). In Table \ref{t1}, $W = {\|G\|_1 \over \|G\|_0}$ is the average weight.

\begin{table}[!h]
  \centering
  \caption{The comparison of existing results (for constant $\epsilon$).}\label{t1}
  \begin{tabular}{|c|c|c|c|c|}
      
      \hline
      \textbf{Method} &  \makecell[c]{\textbf{Additive error for all} \\$(S,V\backslash S)$ \textbf{cuts}} &  \makecell[c]{\textbf{Efficient?}} & \makecell[c]{\textbf{Purely }\\\textbf{additive?}} & \makecell[c]{\textbf{Preserve }\\\textbf{sparsity?}}\\
      \hline 
      Dwork et al. \cite{dwork2006calibrating} & ${O}\left(2^{n/2}\right)$ & No & Yes &  No \\
      \hline
      
      Hardt and Rothblum~\cite{hardt2010multiplicative} &${{O}}\left(n\sqrt{|E|W} \cdot \texttt{poly}(\log n )\right)$ & No &Yes& No\\
      \hline 

      Gupta et al. \cite{gupta2012iterative} & 
      ${{O}}\left(\sqrt{n|E|W}\cdot \texttt{poly}(\log n )\right)$ & No & Yes&  No \\
      \hline
      
      Blocki et al. \cite{blocki2012johnson}& ${{O}}\left(n^{1.5} \cdot \texttt{poly}(\log n )\right)$ & Yes & No &  No \\
      \hline
      
      McSherry and Talwar \cite{mcsherry2007mechanism} & ${{O}}(n\cdot \texttt{poly}(\log n ))$ & No & No & Yes \\
      \hline 
      
      Warner \cite{warner1965randomized} & ${{O}}(n^{1.5}\cdot \texttt{poly}(\log n ))$ & Yes & Yes & No  \\
      \hline
      
      Upadhyay et al. \cite{upadhyay2021differentially} & ${{O}}(n^{1.5}\cdot \texttt{poly}(\log n ))$ & Yes & Yes & No \\
      \hline 
      
      Eliavs et al. \cite{eliavs2020differentially} & ${O}(\sqrt{n|E|W}\log ^2n)$ &Yes &Yes&  No \\
      \hline 
     
     Liu et al. \cite{liu2023optimal} & ${O}(\sqrt{n|E|}\log ^3n)$ &Yes &Yes&  No \\
      \hline 
      (Ours) & \textcolor{red}{${{{O}}}(|E|\log n)$} &  \textcolor{red}{Yes} &\textcolor{red}{Yes}& \textcolor{red}{Yes}\\
      \hline 

  \end{tabular}
\end{table}

\subsection{Online algorithms}
Much as our task is to output a synthetic graph that preserves cut approximation, a group of algorithms that we are also interested in answering all cut queries in the interactive setting, in which the algorithm works like a stateful API that handles each cut query. Compared to the offline setting (i.e., output a synthetic graph), the interactive setting is usually considered less demanding (see, \cite{gupta2012iterative,ganev2022robin}). A naive approach in such an online setting is to add an independent Laplace noise on the answer to each cut query. Recall that under the edge level differential privacy, the sensitivity of cut queries is $1$. Therefore, by the advanced composition lemma of differential privacy in Dwork et al. \cite{dwork2010boosting}, to answer $k$ cut queries, we need to add a noise $Z\sim \texttt{Lap}(\sqrt{k}/\epsilon)$ to preserve $(\epsilon,\delta)$-differential privacy, where $\delta$ is negligible. Since there are at most ${O}(2^n)$ different $(S,V\backslash {S})$-cut queries, then the Laplace mechanism has noise proportional to $2^{n/2}/\epsilon$.

The private multiplicative weights method is an important improvement in the online analysis of private datasets (see Hardt and Rothblum \cite{hardt2010multiplicative}). In  Gupta et al.~\cite{gupta2012iterative}, the authors use a multiplicative weights algorithm (as a component of the iterative database construction approach) to answer linear queries normalized in the range $[0,1]$ as we have defined before. Given parameter $\epsilon$, this approach gives an ${O}(\sqrt{n|E|/\epsilon})$ bound on the purely additive error on answering $(S,V\backslash {S})$-cut queries if the input graph is unweighted. However, it's easy to verify that their algorithm can be easily extended to an ${O}(\sqrt{n|E|W/\epsilon})$ upper bound for weight graphs whose average weight is $W$. This is because the error bound relies on the $\ell_1$ norm the vector representation of the input graph, namely the sum of weights.

We note that the multiplicative weights approach answers each cut query efficiently, but it cannot reconstruct a synthetic graph efficiently since there are exponentially many $(S,T)$-cut queries. In our algorithm, we construct a new graph in almost linear time. Note that our algorithms can also answer linear queries with a purely additive error, which implies that our algorithms can approximate both $(S,V\backslash {S})$-cut queries and $(S,T)$-cut queries with a purely additive error like the multiplicative weights approach. Moreover, for sparse weighted graphs where $|E| = {{O}}(n)$, our algorithm outperforms the online algorithms. However, the utility bound is not directly comparable for dense graphs where $|E| = \omega(n)$.

\subsection{Multiplicative Gaussian}

In Blocki et al. \cite{blocki2012johnson}, the authors use
the Johnson-Lindenstrauss transformation on the square root of the Laplacian matrix of graphs to preserve differential privacy. Formally speaking, fix an integer $r$, they pick a random Gaussian matrix $M\in \mathbb{R}^{r\times {n\choose 2}}$ whose entries are independently drawn from $\mathcal{N}(0,1)$, and output a private sketch $\widehat{L}_G = \frac{1}{r} E_G^\top M^\top M E_G$, where $E_G\in \mathbb{R}^{{n\choose 2}\times n}$ is the edge adjacency matrix of the input graph. In the worst case, the multiplicative Gaussian method gives $\widetilde{{O}}(n^{1.5})$ additive error, with a constant multiplicative error (caused by the Johnson-Lindenstrauss transformation).

Compared to our algorithm, the private sketch $\widehat{L}_G$ produced by the multiplicative Gaussian method isn't necessarily a Laplacian matrix of any graph. Moreover, when $|E| = o(n^{1.5})$, the multiplicative Gaussian method pays more additive error.

\subsection{Additive Gaussian/random flipping}
Instead of multiplying a random Gaussian matrix, a much more direct method is to overlay a random Gaussian graph on the original graph $G$. That is, adding a Gaussian noise drawn from $\mathcal{N}(0,{O}(\log(1/\delta)/\epsilon^2))$ on each pair of vertices independently, and releasing the new graph as $\widehat{G}$. By the Gaussian mechanism \citep{dwork2014algorithmic}, we see that it preserves $(\epsilon, \delta)$ differential privacy. For the utility part, since each cut contains at most ${O}(n^2)$ Gaussian noise, by the property of sub-Gaussian, one can verify that for each cut $S$, $|\Phi_G(S) - \Phi_{\widehat{G}}(S)|\leq \widetilde{{O}}(n/\epsilon)$ with high probability. Using the union bound over all $2^n$ different $(S,V\backslash {S})$-cuts, the total error becomes $\max_{S\subset V} |\Phi_G(S) - \Phi_{\widehat{G}}(S)|\leq \widetilde{{O}}(n^{1.5}/\epsilon)$ with high probability. Again, for sparse graphs when $|E| = o(n^{1.5})$, the additive Gaussian mechanism performs worse than our algorithm.

Another noise-adding strategy is to use the randomized response. In particular, for each pair of vertices $\{u,v\}$ of an unweighted graph, we choose its weight to be $1$ or $-1$ independently, with probability $Pr[w_{\{u,v\}}' = 1] = \frac{1+\epsilon w_{\{u,v\}} }{2}$ and $Pr[w_{\{u,v\}}' = -1] = \frac{1-\epsilon w_{\{u,v\}} }{2}$. It is easy to verify that randomized response gives an $\widetilde{{O}}(n^{1.5}/\epsilon)$ purely additive error, which is the same as the additive Gaussian method (see Blocki et al. \cite{blocki2012johnson} for the calculation). The main flaw of the randomized response mechanism is that it only works on unweighted graphs.

\subsection{Exponential mechanism}

A well-known fact in graph theory is that each weighted graph has a cut-sparsifier that approximates the sizes of all $(S,V\backslash {S})$-cuts~\citep{benczur1996approximating}. A graph $G'$ is an $\alpha$ cut-sparsifier for $G$ if it holds that for any $S\subseteq V$, $(1-\alpha)\Phi_G(S)\leq \Phi_{G'}(S) \leq (1+\alpha) \Phi_{G}(S)$. There have been a lot of works that have shown that each weighted graph with non-negative weights has an $\alpha$ edge-sparsifier with at most ${O}(n/\alpha^2)$ edges~\citep{allen2015spectral, batson2012twice, lee2015constructing}. If the maximum weight is polynomial in $n$, we can always describe any such sparse edge-sparsifier in at most $\widetilde{O}(n)$ bits, and thus there are at most $2^{\widetilde{O}(n)}$ such sparsifiers. Therefore, we can consider all edge sparsifiers to be candidates, and as suggested in Blocki et al. \citep{blocki2012johnson}, for any input graph $G$ and sparsifier $H$, we use $$q(G,H) = \max_S \{\min_{\eta: |\eta-1|\leq \alpha} |\Phi_H(S) - \Phi_G(S)| \}$$ as the scoring function. Clearly, such a scoring function has sensitivity $1$ under edge level differential privacy; therefore, is $(\epsilon,0)$-differentially private. Further, using Lemma \ref{l.exp_utility}, it is easy to verify that this instantiation of the exponential mechanism results in an additive error of $\widetilde{{O}}(n)$ with a constant multiplicative error, the latter due to sparsification. However, it is not clear how to sample from this instantiation of the exponential mechanism; while the score function is convex, the space is not convex, and therefore, algorithms such as the ones that give mixing time of sampling from log-concave distribution~\citep{bassily2014private} do not apply. Furthermore, in the exponential mechanism, not only do we need to sample in a space of exponential size, but computing the scoring function for a specified $H$ is also hard. 

\subsection{Mirror descent}
The mirror descent methods optimize a convex function $f(x)$ over a convex domain. Since Liu et al. \cite{liu2023optimal} uses the mirror-descent based algorithm of Eli{\'a}{\v{s}} et al. \cite{eliavs2020differentially}, in what follows, we just discuss Eli{\'a}{\v{s}} et al. \cite{eliavs2020differentially}. In Eli{\'a}{\v{s}} et al. \cite{eliavs2020differentially}, for any input weighted graph $G$ and candidate $G'$, the authors use the cut norm~\cite{alon2004approximating}: 
$$d_{cut}(G,G') = \max\{|x^\top (A-A')y|; x,y\in \{0,1\}^n\}$$
as the loss function over the space of a semi-definite cone of Laplacian matrices. They then use private mirror descent to optimize $G'$. If we use the entropy mirror map, i.e., $\Phi(x) = \sum_{e\in {V \choose 2}} x_e\log x_e$, it can be shown that private mirror descent gives ${O}(\sqrt{n|E|W/\epsilon})$ purely additive error in expectation, which is same as the result of multiplicative weights approach (using entropy regularizer). However, the mirror descent approach outputs a synthetic graph in polynomial time, while the multiplicative weights method doesn't. Thus, this algorithm can be seen as an efficient implementation of {\em private multiplicative weight} (PMW) algorithm.

Compared to the mirror descent approach, we do not need any dependency on the maximum weight $W = \| G\|_\infty$. However, it's still an interesting question whether the square root dependency on $W$ in mirror descent approaches is a fundamental barrier. Recall the scale invariance property (Remark \ref{rem:scale}), the dependency on $W$ in  Eli{\'a}{\v{s}} et al. \cite{eliavs2020differentially} seems necessary due to the trade-off between $\epsilon$ and edge weights. Also, in Lemma 4.3 of  Eli{\'a}{\v{s}} et al. \cite{eliavs2020differentially}, the authors show that the error between the cut norm and its convex approximation has relevance to the sum of weights, and hence $W$ also comes into play.
Another important property of our algorithm is that it preserves the sparsity. That is, if we are promised that the input graph is sparse, then the output graph is also sparse. While in Eli{\'a}{\v{s}} et al. \cite{eliavs2020differentially}, the mirror descent method always outputs a weighted complete graph.

\begin{remark}
    [Discussion regarding high probability bound.] 
\label{rem:highprob}
The bound in Eli{\'a}{\v{s}} et al. \cite{eliavs2020differentially} is in expectation and they refer to Nemirovski et al. \cite{nemirovski2009robust} to get a high probability bound. Nemirovski et al. \cite{nemirovski2009robust} presents two general ways to get a high probability bound from the expectation. The stronger condition (stated as eq. (2.50) in Nemirovski et al. \cite{nemirovski2009robust} and stated below) that allows a bound that gives $1-\beta$ while incurring an extra $\log(1/\beta)$ factor is achieved if 
\[
\mathbb E\left[\exp\left({\|g\|_\infty^2 \over M^2 }\right)\right] = \int n^2 \cdot exp(-t) \cdot exp(t^2/M^2)
\]
is bounded. Here $\|g\|_\infty$ is the $\ell_\infty$ norm of gradients in the mirror descent. However, the algorithm in Eli{\'a}{\v{s}} et al. \cite{eliavs2020differentially} does not satisfy this stronger condition. In fact, the expectation is unbounded if $M$ is poly-logarithmic in $n$. Eli{\'a}{\v{s}} et al. \cite{eliavs2020differentially} satisfies the weaker condition (stated as eq. (2.40) in Nemirovski et al. \cite{nemirovski2009robust}). In fact, Lemma 4.5 in Eli{\'a}{\v{s}} et al. \cite{eliavs2020differentially} implies that $\mathbb E[\|g\|_\infty^2] = O(log^2 n).$ This condition though translates an expectation bound to $1-\beta$ probability bound with an extra ${1 \over \beta^2}$ factor in the error. In particular, if we want $\beta = o(1/\sqrt{n})$, then Eli{\'a}{\v{s}} et al. \cite{eliavs2020differentially} results in an error $O(\sqrt{mn^3/\epsilon} \log^2(n/\delta))$. This is worse than our bound, irrespective of whether the graph is weighted or unweighted or is sparse or dense. 
\end{remark}

\section{Some Real World Examples}
\label{app:usecases}
In this section, we enumerate some of the other real-world examples of graphs that are sparse and heavily weighted, and there are natural reasons to study cut queries. This list is by no means exhaustive, but it serves to show that, in practice, one almost always encounters sparse highly weighted very large graphs in terms of nodes, making $\widetilde O(n^7)$ time algorithm infeasible in practice.

Three of the main applications of graph analysis are in understanding social network graphs, financial transactions, and internet traffic. In what follows, we give representative examples of each of these settings and supplement our claims that these examples result in sparse weighted graphs with publicly available information.

\paragraph{Social network graph} There are four major social network graphs in the current age: three of them are subsidiaries to Facebook: Instagram, WhatsApp, and messenger, and one to Apple: iMessage. For each of these, there is a natural way to construct graphs. 
\begin{itemize}
    \item WhatsApp is an undirected weighted graphs, where the nodes are users and there is an edge if two people interact using WhatsApp with the edge weight representing the number of messages exchanged between them. According to reports, in the year 2021, there were approximately $2$ billion users and about $36$ trillion total messages (or the sum of weights on the edges) were sent only in 2021.  
    \item Instagram is a directed graph, where a user corresponds to an account and there is a directed edge if a user interacts (in terms of likes on the post or comments) with the other user. In 2022, there were about $2.35$ billion (active or inactive) Instagram accounts and just on a single personal (Christiano Ronaldo), there are more than $90$ billion likes. 
    \item Facebook messenger is another undirected graph with about $3$ billion users in 2022 and a total of about $5$ trillion messages exchanged just in the year 2022. 

    \item Finally iMessage has currently 1.3 billion users and total 18.2 trillion messages were exchanged in 2021.  
\end{itemize}

In all these cases, it is easy to verify that the degree of any node (that is other users any user interacts with) is significantly less. Here, the natural notion of privacy is whether, at a given time, two users interact or not. This corresponds to edge-level differential privacy. Moreover, in social network, like Instagram or Messenger, unless an account is private, one can always infer who is connected to whom. Moreover, the companies holding these data already know the connection; only the messages are end-to-end encrypted. 

\paragraph{Financial graphs} We use the example of Chase Bank and transactions on its ATM. According to reports, there are approximately $18$ million accounts in Chase Bank and $16,000$ ATMs, while the total number of ATM transactions done in 2021 is more than $600$ million. To the best of our knowledge, Chase Bank has never published the total number of transactions between its accounts or transactions done through other methods, like Zelle or Venmo.  

\paragraph{Internet Activity Graph} Currently, there are about $4.3$ billion active IP addresses. However, just one search engine, Google, there are about $90,000$ web-search per second. The internet is a sparse graph because the number of connections between nodes (websites, servers, etc.) is much smaller than the total number of possible connections.

\section{An improved upper bound for fixed topologies}\label{app:known_topology}

In this section, we relax edge-level differential privacy to a slightly weaker notion and prove a better utility bound under this setting. 
In short, we assume the topology (i.e., the edge set) is already public. This setting aligns with practical scenarios where data analysts already know the connections, one of the use cases that has been elaborated in Appendix \ref{app:usecases}. Two graphs with the same edge set are neighboring if the edge weight $w_e$ differs by $1$ for exactly one $e\in E$. This assumption has also been well studied in the literature on private graph analysis (see \cite{bodwin2024discrepancy, chen2023differentially, fan2022distances, sealfon2016shortest, roth2021mycelium} for example). In this setting, we assume $|E| = \Omega(n)$ without losing generality. Otherwise, we can always remove the isolated vertices since the topology is known.

Given that we don't have to achieve privacy on topology, we just need to add independent Laplace noise on each edge in $E$: 
\begin{enumerate}
  \item For any $e\in E$, draw an independent Laplace noise $Z\sim \texttt{Lap}(1/\epsilon)$;
  \item Let $w_e \leftarrow w_e+Z.$
\end{enumerate}

The application of the Laplace mechanism yields the following theorem for the utility of private cut approximation:

\begin{theorem}\label{t.topology_is_known}
  Given $\epsilon>0$, there is a $(\epsilon,0)$-differentially private algorithm (under the assumption that the topology is public) such that for any input weighted graph $G = (V,E)$, with high probability, it outputs a $\widehat{G} = (V,\widehat{E})$ such that for any $(S,T)$-cut where $S,T\subset V$ and $S\cap T  = \varnothing$, $\widehat{G}$ satisfies 
  $$|\Phi_G(S,T) - \Phi_{\widehat{G}}(S,T)| \leq \min\left\{O\left(n{\log n}\cdot \sqrt{\min\{|S|,|T|\}}\right), O\left(\sqrt{n|E|\log^2 n}\right)\right\}.$$
  \end{theorem}

\begin{remark}
  We note that under the assumption where the topology is public, the second bound of Theorem \ref{t.topology_is_known} entirely removes the dependency on the maximum edge weight $W$, compared to the $O(\sqrt{n|E|W})$ upper bound proposed in \cite{eliavs2020differentially}. Here, we pay an extra $\log n$ factor in both bounds due to the union bound, while \cite{eliavs2020differentially} only guarantees that the expected accuracy is $O(\sqrt{n|E|W})$.
\end{remark}

Clearly, the algorithm given in this section is $(\epsilon,0)$-differentially private (see Lemma \ref{l.laplace}). We prove the utility part of this theorem by the tail inequality of independent combinations of Laplace noise, which can be proved by standard techniques:
  \begin{lemma}
  [\cite{kotz2001laplace, gupta2012iterative}]
  \label{l.tail_of_laplace}
     Given $k\in \mathbb{N}_{>0}$ and $b>0$. Let $\{Z_i\}_{i\in [k]}$ be i.i.d random variables such that $Z_i\in \texttt{Lap}(b)$, then for any $q_1,\cdots q_k\in [0,1]$, 
    $$Pr\left[\sum_{i\in [k]} q_iZ_i > \alpha\right] \leq \left\{
        \begin{aligned}
            &\exp\left(-\frac{\alpha^2}{6kb^2}\right), \quad &\alpha \leq kb, \\
            &\exp\left(-\frac{\alpha}{6kb}\right), &\alpha > kb.
        \end{aligned}
        \right.$$
    \end{lemma}
\noindent Now, we are ready to prove Theorem \ref{t.topology_is_known}.
\begin{proof}
  [Proof of Theorem \ref{t.topology_is_known}] 
   
We first give the bound $O\left({n}\log n\sqrt{\min\{|S|,|T|\}}\right)$. For any set of vertices $S$, let $E(S,T)$ be the collection of edges in the given graph that crosses $S$ and $T$. Given $G$ and any disjoint $S,T\in 2^V$ such that $|S|\leq |T|$, the error in the synthetic graph can be written as $\texttt{err}(S,T) = \sum_{e \in E(S,T)} Z_e$, where $\{Z_e\}_{e\in |E|}$ are i.i.d Laplace noise from $\texttt{Lap}(1/\epsilon)$. 

Let $\alpha(S) = \min\{n|S|\cdot \lceil\log n\rceil,{n\choose 2}\}$. One can further find a set of $\alpha(S)$ edges in a complete graph that includes $E(S,T)$, let that set be $\tau(S)$. Then, the error can be re-written as $$\texttt{err}(S,T) = \sum_{e\in \tau(S)} q_e Z_e, 
\quad \text{where} \quad q_e = \begin{cases}
    1 & e\in \tau(S)\cap E(S,T) \\
    0 & \text{otherwise}
\end{cases}.$$ 

By Lemma \ref{l.tail_of_laplace}, we see that 
\begin{equation}
  Pr\left[\sum_{e\in \tau(S)} q_eZ_e > \frac{\sqrt{n\log n|\tau(S)|}}{\epsilon}\right] \leq \exp\left(-\frac{n\log n|\tau(S)|/\epsilon^2}{6|\tau(S)|/\epsilon^2}\right) = \exp\left(-\frac{n\log n}{6}\right)
\end{equation}
since $\frac{\sqrt{n\log n|\tau(S)|}}{\epsilon}  \leq \frac{1}{\epsilon}\min \left\{n\log n\sqrt{|S|}, O(n^{1.5}\sqrt{\log n})\right\} \leq \left(\frac{|\tau(S)|}{\epsilon}\right)$. Since there are at most $2^{2n}$ different $(S,T)$-cuts, then we see that 
\begin{equation}
  \begin{aligned}
    &Pr\left[\exists \text{ disjoint }S,T\in 2^V \texttt{ s.t.} |S|\leq |T| \land \texttt{err}(S,T) > n\log n\sqrt{|S|}\right] \\
    &\leq Pr\left[\exists \text{ disjoint }S,T\in 2^V \texttt{ s.t.} |S|\leq |T| \land \texttt{err}(S,T) > \sqrt{n\log n|\tau(S)|}\right]\\
    &\leq \exp\left(-\Theta(\log n)\right) = O\left(\frac{1}{\texttt{poly}(n)}\right).
  \end{aligned}
\end{equation}
The case where $|S|\geq |T|$ is symmetric. Thus, we have that with probability at least $1-O(1/n^c)$, the error is at most $O\left({n}\log n\sqrt{\min\{|S|,|T|\}}\right)$. 

The argument for the second bound is almost the same, in which we just replace the size of $\tau'(S)$ by $\min\{|E|\cdot \lceil \log n\rceil, {n\choose 2}\}$. This is because that $|E|\geq |E(S,T)|$, thus we can still find such $\tau'(S)$ that includes $E(S,T)$. Then, we have that 
$$\frac{\sqrt{n\log n\cdot |\tau'(S)|}}{\epsilon}  \leq \frac{1}{\epsilon}\min\{\log n\cdot \sqrt{n|E|} ,O(n^{1.5}\sqrt{\log n})\} \leq\left(\frac{|\tau'(S)|}{\epsilon}\right)$$
since we have assumed that $|E| = \Omega(n)$. Similarly, we have that 
\begin{equation}
  \begin{aligned}
    &Pr\left[\exists \text{ disjoint }S,T\subseteq V \texttt{ s.t.} |S|\leq |T| \land \texttt{err}(S,T) > \log n\sqrt{n|E|}\right] \\
    &\leq Pr\left[\exists \text{ disjoint }S,T\subseteq V \texttt{ s.t.} |S|\leq |T| \land \texttt{err}(S,T) > \sqrt{n\log n|\tau'(S)|}\right]\\
    &\leq \exp\left(-\Theta(\log n)\right) = O\left(\frac{1}{\texttt{poly}(n)}\right),
  \end{aligned}
\end{equation}
which is what is claimed in Theorem~\ref{t.topology_is_known}.
\end{proof}

\end{document}